\documentclass[10pt,journal,compsoc]{IEEEtran}

\usepackage[nocompress]{cite}

\usepackage{xcolor}
\usepackage[colorlinks=true,citecolor=blue]{hyperref}
\usepackage{graphicx}
\usepackage{amsmath,amsthm,amssymb}
\usepackage{pbalance}
\usepackage{multirow}
\usepackage{algorithm}
\usepackage{algorithmic}
\usepackage{url}
\usepackage{tikz}
\usepackage{comment}

\usetikzlibrary{arrows.meta, positioning, shapes.geometric, fit, backgrounds, calc}

\floatname{algorithm}{Procedure}

\newtheorem{theorem}{Theorem}
\newtheorem{proposition}[theorem]{Proposition}

\newtheorem{definition}[theorem]{Definition}
\newtheorem{remark}[theorem]{Remark}
\newtheorem{example}[theorem]{Example}

\begin{document}

\title{$p$-adic Bi-Filtrations for Topological Machine Learning on Genomic Sequences}

\author{Tirtharaj~Dash$^{1,*}$
        and~Gunja~Sachdeva$^{2}$
\IEEEcompsocitemizethanks{
\IEEEcompsocthanksitem $^{1}$Department of CS \& IS, BITS Pilani, K K Birla Goa Campus, Zuarinagar, Goa 403726, India. E-mail:~\href{mailto:tirtharaj@goa.bits-pilani.ac.in}{tirtharaj@goa.bits-pilani.ac.in}
\IEEEcompsocthanksitem $^{2}$Department of Mathematics, BITS Pilani, K K Birla Goa Campus, Zuarinagar, Goa 403726, India. E-mail:~\href{mailto:gunjas@goa.bits-pilani.ac.in}{gunjas@goa.bits-pilani.ac.in}
\IEEEcompsocthanksitem $^{*}$Corresponding author
}}

\IEEEtitleabstractindextext{%
\begin{abstract}
We introduce \texttt{pVR}, a topological machine learning framework for alignment-free genomic sequence classification that combines $p$-adic numbers with topological data analysis. 
Each DNA sequence is encoded along two complementary axes: a $p$-adic distance on $k$-mer prefixes, which captures hierarchical positional structure, and a compositional $L_1$ distance on $k$-mer frequencies, which captures local sequence content. 
The two distances jointly parameterise a bi-filtered Vietoris--Rips complex, and per-sequence topological summaries from this bi-filtration serve as features for standard machine learning classifiers. 
We establish theoretical guarantees for the construction: stability under metric perturbations and invariance to the choice of prime, alongside a result that explains why a single $p$-adic axis is topologically uninformative and why the bi-filtration recovers nontrivial homology. 
On twelve genomic benchmarks ($28$ to $500$ sequences, $3$ to $7$ classes), 
\texttt{pVR} outperforms four established alignment-free baselines on three of six
low-sample datasets, with gains of up to $21$ percentage points;
it underperforms only on a SARS-CoV-2 variant benchmark whose point-mutation divergence violates the hierarchical assumption, and all methods saturate in the large-sample regime. 
\texttt{pVR} also outperforms zero-shot frozen embeddings from the 500M-parameter Nucleotide Transformer~v2 by $6.7$ to $11.4$ percentage points on three low-sample benchmarks. 
\texttt{pVR} codebase is publicly available at \url{https://github.com/MAHI-Group/pVR}.
\end{abstract}

}

\maketitle

\IEEEdisplaynontitleabstractindextext
\IEEEpeerreviewmaketitle

\section{Introduction}
\label{sec:introduction}

Comparative genomic analysis is a foundational task in bioinformatics. 
It includes classifying organisms, identifying variants, and reconstructing evolutionary relationships from genetic sequences. Classical alignment-based approaches such as MAFFT and MUSCLE are effective for closely related sequences but scale poorly with sequence length and divergence. Alignment-free methods address these limitations by mapping sequences to fixed-length feature vectors, enabling scalable comparison without explicit alignment~\cite{zielezinski2017alignment}. Among these, $k$-mer frequency methods~\cite{sims2009alignment} compare distributions of short subsequences, the natural vector method (NVM)~\cite{deng2011novel} encodes positional statistics of nucleotides, and MinHash-based tools like Mash~\cite{ondov2016mash} estimate Jaccard similarity via locality-sensitive hashing. 
While these methods are effective, they treat $k$-mers as independent entities and discard the structural relationships that exist among them.

Recently, topological data analysis (TDA) has emerged as a tool for capturing higher-order structures in genomic data. 
Chan et al.~\cite{chan2013topology} applied persistent homology to detect reassortment in viral evolution. Hozumi and Wei~\cite{hozumi2024revealing} introduced $k$-mer topology using persistent Laplacians, and Suwayyid et al.~\cite{suwayyid2025cakl} proposed commutative algebra $k$-mer learning (CAKL) using persistent Stanley--Reisner invariants. Both demonstrate that topological and algebraic invariants capture information invisible to frequency-based approaches.

A separate line of work has used $p$-adic distances to model biological sequences. Dragovich and Dragovich~\cite{dragovich2010p} showed that encoding nucleotides as digits in $\mathbb{Z}_5$ yields a metric space where codon degeneracy corresponds to $p$-adic proximity. Dragovich et al.~\cite{dragovich2021p} surveyed broader connections between $p$-adic analysis, ultrametric spaces, and models of genetic codes. 
Unlike ordinary metrics, $p$-adic distances are ultrametric, satisfying the strong triangle inequality $d_p(x,z) \leq \max(d_p(x,y), d_p(y,z))$. 
Finite ultrametric spaces are exactly the leaf sets of rooted weighted trees~\cite{carlsson2010characterization}, which is the precise sense in which the $p$-adic distance imposes a hierarchical, tree-like structure on $k$-mers. This structure may align with evolutionary divergence when sequence variation accumulates along tree-like lineages~\cite{semple2003phylogenetics}.

Despite these parallel developments, to the best of our knowledge, no method has combined $p$-adic encodings with simplicial persistent homology for genomic sequence classification. 
In this work, we bring these two ideas together for genomic sequence classification.
We first provide self-contained introductions to $p$-adic distance and Vietoris--Rips complexes in Section~\ref{sec:prelim} for readers without specialised backgrounds. 
Our contributions are both theoretical and empirical. The first three results are
structural: in the idealised setting of a strictly ultrametric distance, they explain
why a single $p$-adic axis cannot capture higher-order topology while pairing it with a
compositional axis can. Since the implemented distance $D_p$ is only approximately
ultrametric (Remark~\ref{rem:strict_vs_empirical}), these results motivate \texttt{pVR}
rather than describe it exactly, and we confirm the predicted behaviour empirically
(Section~\ref{sec:results}). The fourth contribution implements this design for genomic
classification. 
We state these contributions precisely below.
\begin{enumerate}
\item We prove that Vietoris--Rips complexes built from finite ultrametric spaces have
trivial homology in all positive dimensions (Theorem~\ref{thm:trivial}), so a
single strictly ultrametric axis cannot capture higher-order structure. This motivates a
second filtration axis.
\item We introduce a bi-filtered Vietoris--Rips complex using both $p$-adic and
compositional $L_1$ distances and prove, on a strictly ultrametric configuration, that it
recovers nontrivial homology absent from either filtration alone
(Proposition~\ref{prop:nontrivial}).
\item We prove that the resulting bi-persistence module is stable under metric perturbations (Proposition~\ref{prop:stability}) and invariant to the choice of prime (Proposition~\ref{prop:prime_invariance}), so the features are robust and free of a prime hyperparameter.
\item We implement these ideas as \texttt{pVR}, an alignment-free framework for genomic sequence classification, and evaluate it on twelve genomic benchmarks spanning two scale regimes. In the low-sample regime, 
pVR outperforms four established alignment-free baselines on several datasets, with gains
of up to $21$ percentage points; in the large-sample regime it remains competitive, where most methods approach saturated performance.
\end{enumerate}

\section{Related Work}
\label{sec:related}
\subsection{Alignment-Free Sequence Comparison}
Most alignment-free methods reduce a sequence to a fixed-length vector and compare these vectors directly, differing mainly in what they put in the vector. The Feature Frequency Profile (FFP) framework~\cite{sims2009alignment} compares $k$-mer count distributions under Jensen--Shannon divergence. The natural vector method (NVM)~\cite{deng2011novel} instead summarises each nucleotide type by its count, mean position, and normalised second central moment. 
Sketch-based tools such as Mash~\cite{ondov2016mash} avoid explicit vectors altogether, using MinHash to approximate the Jaccard similarity between $k$-mer sets; see~\cite{zielezinski2017alignment} for a broader survey. Despite their differences, these methods all treat $k$-mers as independent features and make no use of the hierarchical arithmetic structure of the genetic encoding. 
A more recent line of work instead learns representations directly: 
Large pretrained models such as DNABERT~\cite{ji2021dnabert}, Nucleotide Transformer~\cite{dalla2025nucleotide}, and GROVER~\cite{sanabria2024dna} apply transformer architectures to large collections of DNA sequences. Notably, the representations these models learn remain largely compositional; GROVER's token embeddings, for instance, primarily encode $k$-mer frequency, sequence content, and length~\cite{sanabria2024dna}.

\subsection{Topological Data Analysis in Genomics}
TDA was applied in genomics through the work of Chan et al.~\cite{chan2013topology}, who used Vietoris--Rips persistent homology to detect reassortment events in influenza and HIV; Camara et al.~\cite{camara2017topological} give a broader overview of topological methods in evolutionary biology. Closest to our setting are two recent methods: $k$-mer topology~\cite{hozumi2024revealing}, which applies persistent Laplacians to $k$-mer frequency data, and CAKL~\cite{suwayyid2025cakl}, which extracts algebraic invariants from $k$-mer complexes via persistent Stanley--Reisner theory. Both build simplicial complexes from sequences and read off topological or algebraic features for classification. Topological summaries have also been used to study recombination, evolutionary structure, and sequence organisation more generally~\cite{rabadan2019tdabook}. 
Our construction differs in where the complex comes from. The ultrametric that drives it is a biologically motivated $p$-adic encoding rather than a frequency or sketch distance, 
and we combine it with a compositional metric in a single bi-filtered complex, so the resulting topology reflects how the two distances interact.

\subsection{$p$-adic Models in Biology}
The use of $p$-adic numbers in biology goes back to Dragovich and Dragovich~\cite{dragovich2010p}, who mapped nucleotides to $5$-adic digits and showed that the degeneracy of the genetic code corresponds to $p$-adic proximity of codons. Dragovich et al.~\cite{dragovich2021p} later extended this to protein folding as ultrametric energy landscapes and gene expression as $2$-adic dynamical systems, with an ultrametric similarity for DNA, RNA, and protein sequences sketched in~\cite{dragovich2017ultrametrics}. The idea has since moved into machine learning, in $p$-adic clustering of single-cell RNA-seq data~\cite{sharma2025p}, van der Put neural networks operating natively in $\mathbb{Z}_p$~\cite{n2025v}, and theoretical accounts of $p$-adic classification and representation learning~\cite{martins2025learning}. We find the combination of $p$-adic encodings with simplicial persistent homology to be a largely underexplored direction, particularly for alignment-free genomic sequence classification, and we explore it in this work.

\section{Theoretical Framework}
\label{sec:math}

\subsection{Preliminaries}
\label{sec:prelim}

We first provide the mathematical machinery underlying \texttt{pVR}: the $p$-adic distance, its use in encoding $k$-mers, the Vietoris--Rips complex, and the two sequence-level distances that drive the bi-filtration. We attempt to illustrate each object with worked examples on biological sequence data. 
Throughout the paper we work with a finite set $\mathcal{S} = \{s_1, \ldots, s_N\}$ of DNA sequences over the alphabet $\Sigma = \{A, C, G, T\}$. Readers familiar with $p$-adic numbers and persistent homology may skip to the theoretical results in Section~\ref{sec:theo_results}; general introductions to persistent homology and topological data analysis can be found in~\cite{ghrist2008barcodes}.

\smallskip
\noindent\textbf{$p$-adic distance.}\enspace The $p$-adic distance is an alternative notion of closeness between integers based on divisibility rather than absolute difference. Fix a prime $p$. For any nonzero integer $n$, write $n = p^v \cdot m$ where $\gcd(m, p) = 1$. The exponent $v$ is the \emph{$p$-adic valuation}, denoted $v_p(n)$. The $p$-adic distance between two integers $a, b$ is:
\begin{equation}\label{eq:padic_def}
  d_p(a, b) = p^{-v_p(a - b)},
\end{equation}
with $d_p(a, a) = 0$. Two integers are $p$-adically close when their difference is highly divisible by $p$.

\begin{example}[$5$-adic distance between integers]\label{ex:padic_int}
Let $p = 5$. Then $d_5(7, 3) = 5^0 = 1$ since $7 - 3 = 4$ is not divisible by $5$, while $d_5(7, 132) = 5^{-3} = 0.008$ since $7 - 132 = -125 = -5^3$. Counterintuitively, $132$ is $5$-adically much closer to $7$ than $3$ is: the $p$-adic distance privileges divisibility over magnitude.
\end{example}

\smallskip
\noindent\textbf{Ultrametric property.}\enspace The $p$-adic distance satisfies the standard metric axioms and a strengthened form of the triangle inequality,
\begin{equation}
  d_p(a, c) \leq \max\bigl(d_p(a, b), d_p(b, c)\bigr),
\end{equation}
which is called the \textit{ultrametric inequality}. Geometrically, every triangle in $p$-adic space is isosceles, with the two longest sides having equal length. This forces a hierarchical, tree-like structure on any finite subset, and this constraint underlies Theorem~\ref{thm:trivial} below.

\smallskip
\noindent\textbf{The space $\mathbb{Q}_p$.}\enspace The $p$-adic numbers $\mathbb{Q}_p$ are the completion of $\mathbb{Q}$ under $d_p$, with the $p$-adic integers $\mathbb{Z}_p \subset \mathbb{Q}_p$ given by elements expressible as $\sum_{i \geq 0} a_i \, p^i$, $a_i \in \{0, 1, \ldots, p-1\}$. Finite truncations of these series correspond exactly to non-negative integers written in base $p$. This base-$p$ representation is the basis for the $k$-mer encoding introduced next.

\smallskip
\noindent\textbf{$p$-adic encoding of $k$-mers.}\enspace Following~\cite{dragovich2010p}, we fix a prime $p \geq |\Sigma| + 1 = 5$ and define the digit assignment $\phi: \Sigma \to \{1, 2, 3, 4\}$ by $\phi(A) = 1$, $\phi(C) = 2$, $\phi(G) = 3$, $\phi(T) = 4$. A $k$-mer $w = w_0 w_1 \cdots w_{k-1}$ is encoded as the $p$-adic integer
\begin{equation}\label{eq:encoding}
  \phi_p(w) = \sum_{i=0}^{k-1} \phi(w_i) \cdot p^i,
\end{equation}
treating its nucleotides as digits in base $p$ with the first nucleotide in the lowest-order position. Under $d_p$, two $k$-mers are close exactly when they share a long common prefix, and differences at later positions are absorbed by higher powers of $p$. We show this in an example below.

\begin{example}[$5$-adic distance between $k$-mers]\label{ex:padic_kmer}
With $p = 5$ and the encoding $A \mapsto 1$, $C \mapsto 2$, $G \mapsto 3$, $T \mapsto 4$, consider the $4$-mer $w = \mathtt{ACGT}$. We compute $\phi_5(w) = 1 + 2 \cdot 5 + 3 \cdot 25 + 4 \cdot 125 = 586$. A mutation at the last position, $w' = \mathtt{ACGA}$, gives $\phi_5(w') = 211$ and $\phi_5(w) - \phi_5(w') = 375 = 3 \cdot 5^3$, so $v_5 = 3$ and $d_5(w, w') = 5^{-3} = 0.008$. A mutation at the first position, $w'' = \mathtt{TCGT}$, gives $\phi_5(w'') = 589$ and $\phi_5(w) - \phi_5(w'') = -3$, so $v_5 = 0$ and $d_5(w, w'') = 1$. Thus $p$-adic distance treats positional information asymmetrically: mutations at conserved (early) positions are weighted exponentially more than mutations at variable (late) positions.
\end{example}

\smallskip
\noindent\textbf{Vietoris--Rips complexes.}\enspace A \emph{simplex} generalises a triangle. A $0$-simplex is a point, a $1$-simplex is an edge, a $2$-simplex is a filled triangle, a $3$-simplex is a filled tetrahedron, and so on. A \emph{$k$-simplex} on $k+1$ points is a single combinatorial object recording that those $k+1$ points are mutually connected. A \emph{simplicial complex} $K$ on a finite vertex set $V$ is a collection of simplices satisfying the closure property: if a simplex $\sigma \in K$, every nonempty subset of $\sigma$ is also in $K$. Hence if a triangle $\{a, b, c\} \in K$, then all three of its edges and vertices must also be in $K$. Given a finite metric space $(X, d)$ and a threshold $r \geq 0$, the \emph{Vietoris--Rips complex} $\mathrm{VR}(X; r)$ is the simplicial complex whose simplices are subsets of $X$ in which every pair of points is within distance $r$,
\begin{equation}\label{eq:vr_def}
  \mathrm{VR}(X; r) = \bigl\{ \sigma \subseteq X :
  d(x, y) \leq r \ \forall\, x, y \in \sigma \bigr\}.
\end{equation}
This is a \emph{flag complex}: $\sigma$ is a simplex whenever all its vertex pairs are edges. Because increasing $r$ only adds simplices, the complexes form a \emph{filtration} $\mathrm{VR}(X; r_1) \subseteq \mathrm{VR}(X; r_2) \subseteq \cdots$ for $r_1 \leq r_2 \leq \cdots$.

\begin{example}[Building a small Vietoris--Rips complex]\label{ex:vr_small}
Let $X = \{a, b, c, d\}$ with pairwise distances $d(a,b) = 0.2$, $d(a,c) = 0.5$, $d(a,d) = 0.7$, $d(b,c) = 0.3$, $d(b,d) = 0.6$, and $d(c,d) = 0.4$. Figure~\ref{fig:vr_evolution} illustrates the filtration. At $r = 0.1$ only the four vertices are present and no edges are drawn. At $r = 0.35$ edges $\{a,b\}$ and $\{b,c\}$ are present, but the triple $\{a,b,c\}$ is not yet a triangle because $d(a,c) = 0.5 > 0.35$. At $r = 0.5$ the edges $\{a,b\}, \{b,c\}, \{a,c\}, \{c,d\}$ are present, and the triangle $\{a,b,c\}$ now appears. The complex grows from disconnected points through a graph-like intermediate stage into a fully connected high-dimensional simplex.

\begin{figure}[!htb]
\centering
\begin{tikzpicture}[scale=0.8, every node/.style={font=\footnotesize}]

\def\pa{(0,1.6)}
\def\pb{(1.4,1.4)}
\def\pc{(0.7,0.4)}
\def\pd{(2.0,0.6)}

\foreach \shift/\r/\label in {0/0.1/{$r{=}0.1$}, 4/0.35/{$r{=}0.35$}, 8/0.5/{$r{=}0.5$}}{
  \begin{scope}[xshift=\shift cm]
    \node at (1.3,2.3) {\label};
    \fill[blue!40!black] \pa circle (1.5pt) node[above left] {$a$};
    \fill[blue!40!black] \pb circle (1.5pt) node[above right] {$b$};
    \fill[blue!40!black] \pc circle (1.5pt) node[below] {$c$};
    \fill[blue!40!black] \pd circle (1.5pt) node[right] {$d$};
  \end{scope}
}

\begin{scope}[xshift=4cm]
  \draw[blue!50!black, thick] \pa -- \pb;
  \draw[blue!50!black, thick] \pb -- \pc;
\end{scope}

\begin{scope}[xshift=8cm]
  \fill[blue!20, opacity=0.7] \pa -- \pb -- \pc -- cycle;
  \draw[blue!50!black, thick] \pa -- \pb;
  \draw[blue!50!black, thick] \pb -- \pc;
  \draw[blue!50!black, thick] \pa -- \pc;
  \draw[blue!50!black, thick] \pc -- \pd;
\end{scope}

\node at (1.3,-0.4) {\footnotesize $\beta_0{=}4,\ \beta_1{=}0$};
\node at (5.3,-0.4) {\footnotesize $\beta_0{=}2,\ \beta_1{=}0$};
\node at (9.3,-0.4) {\footnotesize $\beta_0{=}1,\ \beta_1{=}0$};

\end{tikzpicture}
\caption{Three snapshots of a Vietoris--Rips filtration on four points $\{a,b,c,d\}$ as the threshold $r$ increases. At $r = 0.1$ no edges are present and the complex consists of four isolated vertices ($\beta_0 = 4$). At $r = 0.35$ two edges have appeared but no triangle, so the complex is a path graph ($\beta_0 = 2$). At $r = 0.5$ the triangle $\{a,b,c\}$ has formed and the edge $\{c,d\}$ connects the fourth point ($\beta_0 = 1$).}
\label{fig:vr_evolution}
\end{figure}
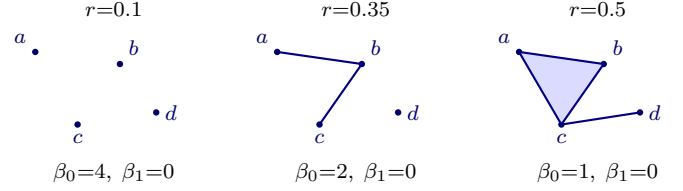

\end{example}

\smallskip
\noindent\textbf{Betti numbers.}\enspace Persistent homology tracks how the shape of $\mathrm{VR}(X; r)$ evolves with $r$ via integer invariants called Betti numbers $\beta_0, \beta_1, \beta_2, \ldots$. Informally, $\beta_0$ counts connected components, $\beta_1$ counts independent $1$-dimensional loops (cycles that are not boundaries of filled triangles), and $\beta_2$ counts independent enclosed voids, and so on. Tracking how these numbers change as $r$ varies is the main object of study in persistent homology~\cite{chan2013topology,edelsbrunner2002topological,zomorodian2004computing}. Applied to sequence data, each sequence $s_i \in \mathcal{S}$ is a vertex, edges encode pairwise similarity below threshold $r$, and higher-dimensional simplices encode groups of mutually-similar sequences; persistent homology then extracts how these groups merge, form loops, and dissolve as we relax the threshold, capturing information that no single distance threshold could reveal. The novelty of our work is to use \emph{two} distances simultaneously, producing the bi-filtration $\mathrm{VR}(\epsilon_p, \epsilon_c)$ defined below, from which we extract per-sequence topological features that combine hierarchical and compositional information for downstream classification.

\smallskip
\noindent\textbf{Sequence-level distances.}\enspace We now lift the $p$-adic distance on $k$-mers, and the standard compositional distance on $k$-mer frequency vectors, to distances on the sequence set $\mathcal{S}$; these two sequence-level distances furnish the two axes of the bi-filtration. At scale $j$ ($j = 1, \ldots, J$ where $J = \min(k, 3)$), we group $k$-mers by their residue $\phi_p(w) \bmod p^j$ and form a normalised frequency vector $h_j^{(s)}$ over $p^j$ bins for each sequence $s \in \mathcal{S}$. Scale $j$ groups $k$-mers sharing the same first $j$ characters, providing a hierarchy from coarse ($j=1$, $p$ bins) to fine ($j=J$, $p^J$ bins). The sequence-level $p$-adic distance is a weighted $L_1$ distance,
\begin{equation}\label{eq:dp}
  D_p(s_i, s_\ell) = \frac{\sum_{j=1}^{J} j \cdot \|h_j^{(s_i)} - h_j^{(s_\ell)}\|_1}{\sum_{j=1}^{J} j}.
\end{equation}
Each scale-$j$ term $\|h_j^{(s_i)} - h_j^{(s_\ell)}\|_1$ is an $L_1$ distance on histograms, hence a metric; a positive weighted sum of metrics is a metric, so $D_p$ is a metric. From an empirical perspective, the cap at $j=3$ keeps the histogram dimension ($p + p^2 + p^3 = 155$ for $p = 5$) comparable to the $4^k = 256$ compositional features at $k = 4$, while avoiding the sparse-bin regime in which $L_1$ distance approximates discrete set-difference rather than smooth statistics.

The compositional distance is simpler. 
For each sequence $s \in \mathcal{S}$, let $f^{(s)} \in \mathbb{R}^{|\Sigma|^k}$ be the vector of normalised $k$-mer frequencies, with components $f^{(s)}(w) = \mathrm{count}(w; s) / m_s$, where $\mathrm{count}(w; s)$ is the number of occurrences of $k$-mer $w$ in $s$ and $m_s$ is the total number of overlapping $k$-mers extracted from $s$. The \emph{compositional distance} between two sequences is the $L_1$ distance between these frequency vectors,
\begin{equation}\label{eq:dc}
  D_c(s_i, s_\ell) = \| f^{(s_i)} - f^{(s_\ell)} \|_1.
\end{equation}
$D_p$ captures hierarchical prefix structure across scales, while $D_c$ captures local compositional content; the two are complementary by design.

\subsection{Theoretical Results}
\label{sec:theo_results}

We now provide the results underlying the construction of \texttt{pVR}. We first show why a single axis is not sufficient (Theorem~\ref{thm:trivial}), then characterise the bi-filtration that follows (Propositions~\ref{prop:nontrivial}--\ref{prop:prime_invariance}).

\smallskip
\noindent\textbf{Triviality of ultrametric VR complexes.}\enspace We begin by recalling the definition of an ultrametric space.

\begin{definition}[Ultrametric space]\label{def:ultrametric}
A metric space $(X, d)$ is \emph{ultrametric} if for all $x, y, z \in X$, $d(x, z) \leq \max(d(x, y), d(y, z))$.
\end{definition}

The following theorem reflects the collapse of higher-order topology in Vietoris--Rips filtrations constructed from ultrametric spaces, which are closely related to hierarchical clustering structures~\cite{carlsson2010characterization}.

\begin{theorem}[Trivial higher homology]
\label{thm:trivial}
Let $(X, d)$ be a finite ultrametric space. For every $r \geq 0$, $\mathrm{VR}(X; r)$ is homotopy equivalent to a discrete set indexed by its connected components. In particular, $H_k(\mathrm{VR}(X; r)) = 0$ for all $k \geq 1$, and $\beta_0(\mathrm{VR}(X; r))$ equals the number of clusters in the single-linkage dendrogram of $(X, d)$ cut at height $r$.
\end{theorem}

\begin{proof}
Let $C$ be a connected component of the $1$-skeleton of $\mathrm{VR}(X; r)$. Take any $x, z \in C$ connected by a path $x = v_0, v_1, \ldots, v_m = z$ with $d(v_{i-1}, v_i) \leq r$. We prove $d(x, z) \leq r$ by induction on $m$. The base case $m = 1$ is immediate. For $m \geq 2$, the inductive hypothesis gives $d(x, v_{m-1}) \leq r$, and $d(v_{m-1}, z) \leq r$ by assumption. By the ultrametric inequality, $d(x, z) \leq \max(d(x, v_{m-1}), d(v_{m-1}, z)) \leq r$. Hence $\{x, z\}$ is an edge, so $C$ is a complete graph. Since the VR complex is a flag complex, the subcomplex on $C$ is the full simplex $\Delta^{|C|-1}$, which is contractible. Thus $\mathrm{VR}(X; r)$ is a disjoint union of contractible simplices.
\end{proof}

\begin{remark}
\label{rem:strict_vs_empirical}
Theorem~\ref{thm:trivial} implies that applying standard single-parameter persistent homology to any finite ultrametric space, for instance, a set of $k$-mers under the $p$-adic distance $d_p$ of Eq.~\ref{eq:padic_def} is uninformative beyond $\beta_0$. This is the fundamental obstacle motivating the bi-filtration. We note one technical caveat. The sequence-level distance $D_p$ in Eq.~\ref{eq:dp} is a weighted $L_1$ aggregate of histograms over $p$-adic residues; it inherits hierarchical prefix structure from $d_p$ but satisfies only the ordinary triangle inequality, not strict ultrametricity. As a concrete example, consider the sequences $a = \texttt{A}^L, \qquad b = \texttt{A}^{L/2}\texttt{C}^{L/2}, \qquad c = \texttt{C}^L$, with $L \gg k$, so that boundary $k$-mers in $b$ are negligible. The $k$-mer distributions of $a$ and $c$ are each concentrated on a single $k$-mer, while that of $b$ is asymptotically split equally between the two. Consequently, in the limit $L \to \infty$, at every scale $j$, $\|h_j^{(a)} - h_j^{(b)}\|_1 \to 1, \qquad \|h_j^{(b)} - h_j^{(c)}\|_1 \to 1$, while $\|h_j^{(a)} - h_j^{(c)}\|_1 \to 2$. Substituting into Eq.~\eqref{eq:dp} therefore gives, in the same limit, $D_p(a,b) \to 1, \qquad D_p(b,c) \to 1, \qquad D_p(a,c) \to 2$, which violates the ultrametric inequality $D_p(a,c) \le \max(D_p(a,b), D_p(b,c))$. Hence $D_p$ is generally not ultrametric. Theorem~\ref{thm:trivial} therefore does not formally apply to the implemented sequence-level distance $D_p$; it instead motivates the bi-filtration as a structural design principle for sequence data, while the corresponding topological collapse of $D_p$-only complexes is observed empirically (Section~\ref{sec:results}) rather than as a direct corollary.
\end{remark}

\smallskip
\noindent\textbf{Nontrivial topology from the bi-filtration.}\enspace We now define the bi-filtered complex used by \texttt{pVR}.

\begin{definition}[Bi-filtered VR complex]
\label{def:bifilt}
Let $X$ be a finite set equipped with two metrics, $d_p$ (a hierarchical or ultrametric distance) and $d_c$ (a compositional distance). The bi-filtered VR complex at thresholds $(\epsilon_p, \epsilon_c)$ is
\[
  \mathrm{VR}(\epsilon_p, \epsilon_c) = \{\sigma \subseteq X : 
  d_p(x,y) \le \epsilon_p \text{ and } d_c(x,y) \le \epsilon_c 
  \;\forall\, x, y \in \sigma\}.
\]
\end{definition}

The \texttt{pVR} pipeline (Section~\ref{sec:methods}) applies Definition~\ref{def:bifilt} with $X = \mathcal{S}$, $d_p = D_p$ (Eq.~\ref{eq:dp}), and $d_c = D_c$ (Eq.~\ref{eq:dc}). The substitution $d_c = D_c$ is rigorous since both are genuine metrics. The substitution $d_p = D_p$ is approximate: $D_p$ inherits hierarchical structure from the $p$-adic encoding of $k$-mers but is not strictly ultrametric (see Remark~\ref{rem:strict_vs_empirical}). The existence-of-cycle result below (Proposition~\ref{prop:nontrivial}) is established on a strictly ultrametric example to demonstrate that the bi-filtration can recover nontrivial topology in principle; its empirical realisation on real sequence data is documented in Section~\ref{sec:results}.

\begin{proposition}[Nontrivial homology in the bi-filtration]
\label{prop:nontrivial}
Let $(X, d_p, d_c)$ be a finite set with ultrametric $d_p$ and metric $d_c$. There exists a configuration for which $H_1(\mathrm{VR}(\epsilon_p, \epsilon_c)) \neq 0$.
\end{proposition}

\begin{proof}
Let $X = \{a, b, c, d, e\}$ with the following distances: $d_p(a,b) = d_p(c,d) = 1/5; d_p(a,c) = d_p(a,d) = d_p(b,c) = d_p(b,d) = 1; d_p(\cdot,e) = 25 \text{ for every other point}; d_c(a,c) = d_c(b,d) = d_c(a,d) = d_c(b,c) = 0.3; d_c(a,b) = d_c(c,d) = 0.5; d_c(\cdot,e) = 0.3 \text{ for every other point}.$ One verifies that $d_p$ is ultrametric (every triangle is isosceles) and $d_c$ satisfies the triangle inequality. At thresholds $(\epsilon_p, \epsilon_c) = (1, 0.4)$, the bi-filtration excludes all edges incident to $e$ (since $d_p(\cdot, e) = 25$) and excludes $\{a,b\}, \{c,d\}$ (since $d_c = 0.5$). The remaining four edges $\{a,c\}, \{a,d\}, \{b,c\}, \{b,d\}$ form a 4-cycle $a$--$c$--$b$--$d$--$a$. No triangle exists in the complex, so $H_1 \cong \mathbb{Z}$. The $d_c$-only filtration at the same $\epsilon_c = 0.4$ includes all eight cross-cluster and witness edges, producing filling triangles $\{a,c,e\}, \{c,b,e\}, \{b,d,e\}, \{d,a,e\}$ that bound the 4-cycle, hence $H_1 = 0$. The $d_p$-only filtration produces a tetrahedron on $\{a,b,c,d\}$ with $e$ isolated, also yielding $H_1 = 0$. Figure~\ref{fig:bifilt_evolution} shows the three configurations.
\end{proof}

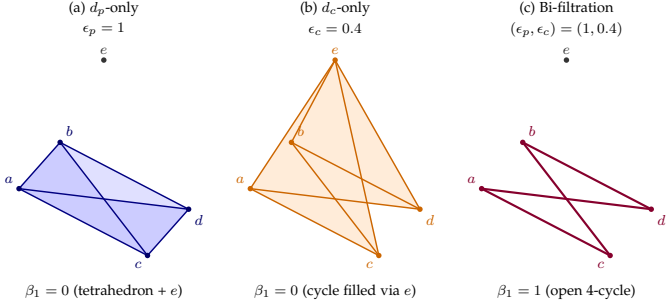
\begin{figure}[!htb]
\centering
\resizebox{\columnwidth}{!}{%
\begin{tikzpicture}[every node/.style={font=\footnotesize}]
\def\pa{(0, 1.5)}
\def\pb{(0.8, 2.4)}
\def\pc{(2.5, 0.2)}
\def\pd{(3.3, 1.1)}
\def\pe{(1.65, 4.0)}

\begin{scope}[xshift=0cm]
  \node at (1.65, 5.0) {(a) $d_p$-only};
  \node at (1.65, 4.6) {\footnotesize $\epsilon_p = 1$};
  \fill[blue!12] \pa -- \pb -- \pd -- \pc -- cycle;
  \fill[blue!20] \pa -- \pb -- \pc -- cycle;
  \fill[blue!20] \pa -- \pc -- \pd -- cycle;
  \draw[blue!50!black, thick] \pa -- \pb;
  \draw[blue!50!black, thick] \pa -- \pc;
  \draw[blue!50!black, thick] \pa -- \pd;
  \draw[blue!50!black, thick] \pb -- \pc;
  \draw[blue!50!black, thick] \pb -- \pd;
  \draw[blue!50!black, thick] \pc -- \pd;
  \fill[blue!40!black] \pa circle (1.5pt) node[above left] {$a$};
  \fill[blue!40!black] \pb circle (1.5pt) node[above right] {$b$};
  \fill[blue!40!black] \pc circle (1.5pt) node[below left] {$c$};
  \fill[blue!40!black] \pd circle (1.5pt) node[below right] {$d$};
  \fill[black!80] \pe circle (1.5pt) node[above] {$e$};
  \node at (1.65, -0.5) {\footnotesize $\beta_1 = 0$ (tetrahedron + $e$)};
\end{scope}

\begin{scope}[xshift=4.5cm]
  \node at (1.65, 5.0) {(b) $d_c$-only};
  \node at (1.65, 4.6) {\footnotesize $\epsilon_c = 0.4$};
  \fill[orange!15] \pa -- \pc -- \pe -- cycle;
  \fill[orange!15] \pc -- \pb -- \pe -- cycle;
  \fill[orange!15] \pb -- \pd -- \pe -- cycle;
  \fill[orange!15] \pd -- \pa -- \pe -- cycle;
  \draw[orange!80!black, thick] \pa -- \pc;
  \draw[orange!80!black, thick] \pa -- \pd;
  \draw[orange!80!black, thick] \pb -- \pc;
  \draw[orange!80!black, thick] \pb -- \pd;
  \draw[orange!80!black, thick] \pa -- \pe;
  \draw[orange!80!black, thick] \pb -- \pe;
  \draw[orange!80!black, thick] \pc -- \pe;
  \draw[orange!80!black, thick] \pd -- \pe;
  \fill[orange!80!black] \pa circle (1.5pt) node[above left] {$a$};
  \fill[orange!80!black] \pb circle (1.5pt) node[above right] {$b$};
  \fill[orange!80!black] \pc circle (1.5pt) node[below left] {$c$};
  \fill[orange!80!black] \pd circle (1.5pt) node[below right] {$d$};
  \fill[orange!80!black] \pe circle (1.5pt) node[above] {$e$};
  \node at (1.65, -0.5) {\footnotesize $\beta_1 = 0$ (cycle filled via $e$)};
\end{scope}

\begin{scope}[xshift=9.0cm]
  \node at (1.65, 5.0) {(c) Bi-filtration};
  \node at (1.65, 4.6) {\footnotesize $(\epsilon_p, \epsilon_c) = (1, 0.4)$};
  \draw[purple!70!black, very thick] \pa -- \pc;
  \draw[purple!70!black, very thick] \pa -- \pd;
  \draw[purple!70!black, very thick] \pb -- \pc;
  \draw[purple!70!black, very thick] \pb -- \pd;
  \fill[purple!70!black] \pa circle (1.5pt) node[above left] {$a$};
  \fill[purple!70!black] \pb circle (1.5pt) node[above right] {$b$};
  \fill[purple!70!black] \pc circle (1.5pt) node[below left] {$c$};
  \fill[purple!70!black] \pd circle (1.5pt) node[below right] {$d$};
  \fill[black!80] \pe circle (1.5pt) node[above] {$e$};
  \node at (1.65, -0.5) {\footnotesize $\beta_1 = 1$ (open 4-cycle)};
\end{scope}
\end{tikzpicture}%
}
\caption{Illustration of Proposition~\ref{prop:nontrivial} on the five-point configuration $\{a, b, c, d, e\}$ at $(\epsilon_p, \epsilon_c) = (1, 0.4)$, where $e$ is compositionally close to every point but $p$-adically far. \textbf{(a)}~$d_p$-only: tetrahedron on $\{a,b,c,d\}$ with $e$ isolated ($\beta_1 = 0$). \textbf{(b)}~$d_c$-only: the 4-cycle $a$--$c$--$b$--$d$--$a$ is filled through $e$ ($\beta_1 = 0$). \textbf{(c)}~Bi-filtration: $e$'s edges are excluded, leaving the 4-cycle unfilled ($\beta_1 = 1$). The cycle is invisible to either single-axis filtration alone.}
\label{fig:bifilt_evolution}
\end{figure}

\begin{remark}
Proposition~\ref{prop:nontrivial} is an existence statement. 
In practice we use the whole 2D Betti landscape, since no single threshold pair carries the signal on its own.
\end{remark}

\begin{example}[Bi-filtration on biological sequences]\label{ex:bifilt}
Consider six DNA sequences forming two evolutionary clades, with representative $4$-mer prefixes shown in the table below. All within-clade pairs share the prefix \texttt{AC} or \texttt{TG}, yielding $d_p \leq 5^{-2} = 0.04$. Cross-clade pairs differ at position $0$ and yield $d_p = 1$. At $\epsilon_p = 0.04$, only within-clade edges appear, and the compositional axis reveals internal clade structure. At $\epsilon_p = 1$, cross-clade edges may appear selectively via compositional, potentially creating $1$-cycles that signal cross-clade similarity invisible to either single filtration alone.

\begin{center}
\begin{tabular}{lll}
\hline
Sequence & Clade & $4$-mer prefix ($\phi_5$) \\
\hline
$s_1$ = \texttt{ACGT}\ldots & I & \texttt{ACGT} ($586$) \\
$s_2$ = \texttt{ACGA}\ldots & I & \texttt{ACGA} ($211$) \\
$s_3$ = \texttt{ACTT}\ldots & I & \texttt{ACTT} ($561$) \\
$s_4$ = \texttt{TGCA}\ldots & II & \texttt{TGCA} ($189$) \\
$s_5$ = \texttt{TGCG}\ldots & II & \texttt{TGCG} ($439$) \\
$s_6$ = \texttt{TGTA}\ldots & II & \texttt{TGTA} ($164$) \\
\hline
\end{tabular}
\end{center}
\end{example}

\smallskip
\noindent\textbf{Stability and prime invariance.}\enspace The stability result below applies to any pair of metrics $(d_p, d_c)$ on a finite set; for the \texttt{pVR} pipeline it is applied with $d_p = D_p$ and $d_c = D_c$. The statement uses three standard objects from multi-parameter persistent homology. First, we write $\mathrm{PH}(\mathrm{VR}_{d_p, d_c})$ for the bi-persistence module obtained by applying simplicial homology degree-wise to the bi-filtration of Definition~\ref{def:bifilt}, viewed as a functor $(\mathbb{R}^2, \leq) \to \mathrm{Vec}_F$ over a fixed field $F$ of coefficients~\cite{carlsson2007theory}. Second, given two persistence modules $M, M'$ over $(\mathbb{R}^2, \leq)$, the \emph{interleaving distance} $d_I(M, M')$ is the infimum over $\delta \geq 0$ such that there exist morphisms shifting each module by $(\delta, \delta)$ that compose to the respective structure maps; we refer to~\cite{lesnick2015theory} for the formal definition. Third, for two metrics $d, d'$ on a finite set $X$, the \emph{sup-norm} is $\|d - d'\|_\infty = \max_{x, y \in X} |d(x, y) - d'(x, y)|$. Setting $(d_p, d_c) = (D_p, D_c)$ and $(d_p', d_c') = (D_p', D_c')$ for two sequence sets that differ by small perturbations of either distance therefore gives robustness of the bi-persistence module to such perturbations.

\begin{proposition}[Stability]\label{prop:stability}
Let $(X, d_p, d_c)$ and $(X, d_p', d_c')$ be two bi-metric structures on the same finite set. Using the $L_\infty$ interleaving distance on $(\mathbb{R}^2, \leq)$ in the sense of~\cite{lesnick2015theory},
\begin{align*}
d_I(\mathrm{PH}(\mathrm{VR}_{d_p, d_c}), \mathrm{PH}(\mathrm{VR}_{d_p', d_c'}))
\leq \max(\|d_p - d_p'\|_\infty, \|d_c - d_c'\|_\infty).
\end{align*}
\end{proposition}
\begin{proof}
Set $\delta_p = \|d_p - d_p'\|_\infty$ and $\delta_c = \|d_c - d_c'\|_\infty$. For any pair $x, y \in X$, $d_p'(x,y) \leq d_p(x,y) + \delta_p$ and likewise for $d_c$, so $\mathrm{VR}_{d_p, d_c}(\epsilon_p, \epsilon_c) \subseteq \mathrm{VR}_{d_p', d_c'}(\epsilon_p + \delta_p, \epsilon_c + \delta_c)$ and the reverse inclusion holds by symmetry. Setting $\delta = \max(\delta_p, \delta_c)$, these inclusions yield a diagonal $\delta$-interleaving of the two bi-persistence modules under the $L_\infty$ shift convention on $\mathbb{R}^2$, analogous to the algebraic stability framework for multi-parameter persistence~\cite{lesnick2015theory,carlsson2007theory}.
\end{proof}

\begin{proposition}[Prime invariance]\label{prop:prime_invariance}
Let $\Sigma$ be a finite alphabet with $|\Sigma| = q$. For any two primes $p_1, p_2 \geq q + 1$, the multi-scale histogram distances $D_{p_1}$ and $D_{p_2}$ are identical on any set of sequences over $\Sigma$, up to a relabeling of histogram bins.
\end{proposition}
\begin{proof}
At scale $j$, the histogram bins are indexed by $\sum_{i=0}^{j-1} \phi(w_i) \cdot p^i$ for $j$-prefixes $(w_0, \ldots, w_{j-1}) \in \{1,\ldots,q\}^j$. Since $p \geq q+1$, this map is injective (digits are $< p$, so the base-$p$ representation is unique), and changing $p$ only relabels the bin indices while preserving the bijection with the prefix multiset. The non-zero bins and their counts are determined by the multiset of $j$-prefixes, which is independent of $p$. The $L_1$ distance depends only on count differences across corresponding bins, hence $D_{p_1} = D_{p_2}$.
\end{proof}

\begin{remark}
For the DNA alphabet ($q = 4$), any prime $p \geq 5$ yields identical features. This eliminates $p$ as a hyperparameter, and we use $p = 5$ in all our experiments.
\end{remark}

\section{The pVR Framework}
\label{sec:methods}

\subsection{Method}

We summarise the steps for constructing the \texttt{pVR} pipeline in Procedure~\ref{alg:pvr} along with a graphical overview of this in Figure~\ref{fig:pipeline}.
Each sequence is encoded along two complementary axes: a
$p$-adic histogram capturing hierarchical $k$-mer prefix structure, and a
$k$-mer frequency vector capturing local composition. Pairwise distances
on these two encodings parameterise a bi-filtered Vietoris--Rips complex on a $G_p \times G_c$ threshold grid:
$G_p$ values of the $p$-adic threshold $\epsilon_p$ and $G_c$ values of the compositional
threshold $\epsilon_c$. 
From each grid point we read off
the per-sequence vertex degree (a per-sequence summary of local
connectivity) and the global Betti numbers $\beta_0, \beta_1$. The
concatenation of degree profiles, $p$-adic histograms, and $k$-mer
frequencies forms the feature vector for a standard classifier.
We consider three ML classifiers: XGBoost~\cite{chen2016xgboost}, SVM and $5$-NN.

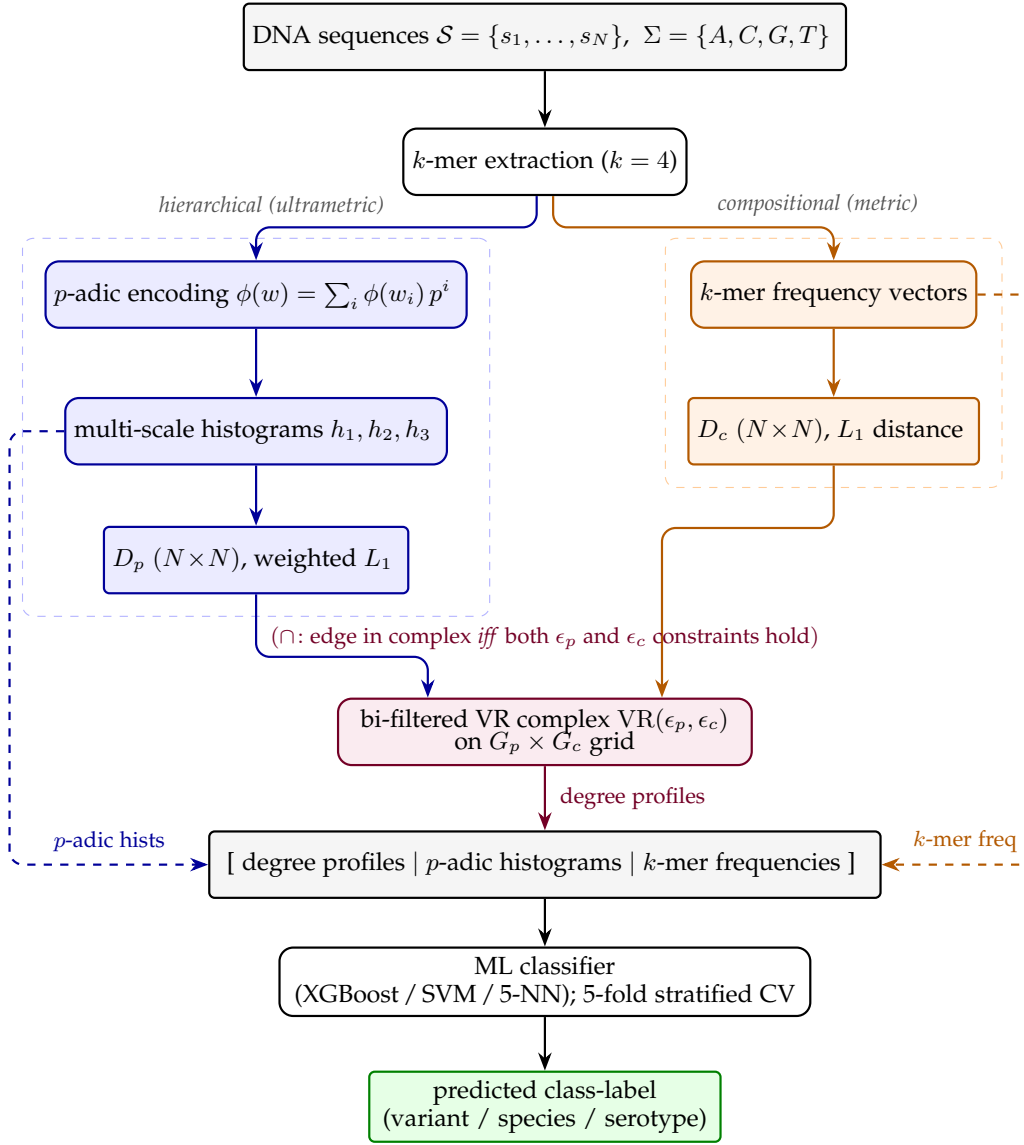
\begin{figure*}[!htb]
\centering
\pgfdeclarelayer{background}
\pgfsetlayers{background,main}
\resizebox{0.75\textwidth}{!}{%
\begin{tikzpicture}[
  font=\small,
  data/.style={
    rectangle, rounded corners=2pt, draw, thick,
    minimum width=3.4cm, minimum height=0.85cm,
    align=center, fill=gray!8
  },
  op/.style={
    rectangle, rounded corners=5pt, draw, thick,
    minimum width=3.4cm, minimum height=0.85cm,
    align=center, fill=white
  },
  padic/.style={fill=blue!8, draw=blue!60!black},
  hamming/.style={fill=orange!10, draw=orange!70!black},
  bifilt/.style={fill=purple!8, draw=purple!60!black},
  final/.style={fill=green!10, draw=green!50!black},
  arrow/.style={->, thick, >=Stealth},
  arrowp/.style={->, thick, >=Stealth, blue!60!black},
  arrowh/.style={->, thick, >=Stealth, orange!70!black},
  arrowbf/.style={->, thick, >=Stealth, purple!60!black},
  lbl/.style={font=\footnotesize\itshape, text=gray!65!black}
]

\node[data, minimum width=7cm] (input) at (0,0) {
  DNA sequences $\mathcal{S}=\{s_1,\ldots,s_N\}$,\;
  $\Sigma=\{A,C,G,T\}$
};

\node[op] (kmer) at (0,-1.6) {$k$-mer extraction ($k=4$)};
\draw[arrow] (input) -- (kmer);

\node[lbl] at (-3.5,-2.15) {hierarchical (ultrametric)};
\node[lbl] at ( 3.5,-2.15) {compositional (metric)};

\node[op, padic] (padic_enc) at (-3.7,-3.3) {
  $p$-adic encoding\;$\phi(w)=\textstyle\sum_i\phi(w_i)\,p^i$
};
\node[op, padic] (padic_hist) at (-3.7,-5.05) {
  multi-scale histograms $h_1,h_2,h_3$
};
\node[data, padic] (Dp) at (-3.7,-6.7) {
  $D_p\;(N\!\times\!N)$, weighted $L_1$
};

\node[op, hamming] (freq) at (3.7,-3.3) {$k$-mer frequency vectors};
\node[data, hamming] (DH) at (3.7,-5.05) {
  $D_c\;(N\!\times\!N)$, $L_1$ distance
};

\draw[arrowp, rounded corners=5pt] ([xshift=-3pt]kmer.south) -- ++(0,-0.4) -| (padic_enc.north);
\draw[arrowh, rounded corners=5pt] ([xshift=3pt]kmer.south) -- ++(0,-0.4) -| (freq.north);

\draw[arrowp] (padic_enc) -- (padic_hist);
\draw[arrowp] (padic_hist) -- (Dp);

\draw[arrowh] (freq) -- (DH);

\node[op, bifilt, minimum width=5.3cm] (bifilt) at (0,-8.9) {
  bi-filtered VR complex $\mathrm{VR}(\epsilon_p,\epsilon_c)$\\[-2pt]
  on $G_p\times G_c$ grid
};

\draw[arrowp, rounded corners=5pt] (Dp.south)  -- ++(0,-0.8) -| ([xshift=-1.5cm]bifilt.north);
\draw[arrowh, rounded corners=5pt] (DH.south)  -- ++(0,-0.8) -| ([xshift=1.5cm]bifilt.north);

\node[font=\footnotesize, text=purple!60!black] at (0,-7.7) {
  $(\cap\!:$ edge in complex \textit{iff} both $\epsilon_p$ and $\epsilon_c$ constraints hold$)$
};

\node[data, minimum width=8.6cm] (features) at (0,-10.6) {
  [\;degree profiles $\mid$ $p$-adic histograms $\mid$ $k$-mer frequencies\;]
};

\draw[arrowbf] (bifilt.south)
  -- node[right, font=\footnotesize, purple!60!black, xshift=2pt] {degree profiles}
  (features.north);

\draw[arrowp, dashed, rounded corners=5pt]
  (padic_hist.west) -- ++(-0.7,0) |- (features.west);
\node[left, font=\footnotesize, blue!60!black] at (-4.75, -10.3) {$p$-adic hists};

\draw[arrowh, dashed, rounded corners=5pt]
  (freq.east) -- ++(0.7,0) |- (features.east);
\node[right, font=\footnotesize, orange!70!black] at (4.6, -10.3) {$k$-mer freq};

\node[op, minimum width=6.8cm] (clf) at (0,-12.1) {
  ML classifier \\
  (XGBoost\,/\,SVM\,/\,5-NN); 5-fold stratified CV
};
\draw[arrow] (features) -- (clf);

\node[data, final, minimum width=4.5cm] (output) at (0,-13.7) {
  predicted class-label \\
  (variant / species / serotype)
};
\draw[arrow] (clf) -- (output);

\begin{pgfonlayer}{background}
  \node[draw=blue!30, dashed, rounded corners, inner sep=8pt,
        fit=(padic_enc)(padic_hist)(Dp)] {};
  \node[draw=orange!40, dashed, rounded corners, inner sep=8pt,
        fit=(freq)(DH)] {};
\end{pgfonlayer}

\end{tikzpicture}
}
\caption{The \texttt{pVR} pipeline. Each sequence passes through two branches: a $p$-adic (hierarchical) branch gives the distance matrix $D_p$, and a compositional ($L_1$) branch gives $D_c$. 
The two matrices parameterise a bi-filtered Vietoris--Rips complex, from which per-sequence degree profiles are extracted; these, with the $p$-adic histograms and $k$-mer frequencies, form the feature vector for a standard classifier. The $\cap$ symbol denotes that an edge appears only when both distance constraints hold.}
\label{fig:pipeline}
\end{figure*}

\begin{algorithm}[!htb]
\caption{\texttt{pVR}: $p$-adic Bi-Filtered Classification}
\label{alg:pvr}
\begin{algorithmic}[1]
\REQUIRE Sequences $\mathcal{S} = \{s_1, \ldots, s_N\}$, labels $y$, prime $p$, $k$-mer size $k$, grid sizes $G_p, G_c$
\ENSURE Predicted class-labels
\STATE Extract overlapping $k$-mers from each $s_i$
\STATE Compute $p$-adic histograms $h_j^{(s_i)}$ at scales $j = 1, \ldots, \min(k,3)$
\STATE Compute $D_p$ via weighted $L_1$ across scales (Eq.~\ref{eq:dp})
\STATE Compute $D_c$ via $L_1$ on $k$-mer frequency vectors
\STATE Choose grid $\{\epsilon_p^{(a)}\}_{a=1}^{G_p} \times \{\epsilon_c^{(b)}\}_{b=1}^{G_c}$
\FOR{each grid point $(\epsilon_p^{(a)}, \epsilon_c^{(b)})$}
  \STATE Construct $\mathrm{VR}(\epsilon_p^{(a)}, \epsilon_c^{(b)})$
  \STATE Compute degree of each vertex $s_i$ in the $1$-skeleton
  \STATE Compute $\beta_0, \beta_1$
\ENDFOR
\STATE \textbf{Features} for each $s_i$: concatenate degree profiles, $p$-adic histograms, $k$-mer frequency vectors
\STATE Standardise features; train and test an ML classifier
\RETURN predicted class-labels, $\hat{y}$
\end{algorithmic}
\end{algorithm}

We now describe these three groups. For each sequence $s_i$:
(1)~\emph{degree profiles}, the number of neighbours of $s_i$ in $\mathrm{VR}(\epsilon_p, \epsilon_c)$ at each of the $G_p \times G_c$ grid points, capturing how local connectivity evolves across the bi-filtration; (2)~\emph{multi-scale $p$-adic histograms}, concatenated across scales; and (3)~\emph{$k$-mer frequency vectors}.
The Betti numbers $\beta_0$ and $\beta_1$ are global properties of the complex at each grid point and are therefore identical across all sequences in a dataset. 
They serve as dataset-level descriptors rather than per-sequence features,
and therefore
the per-sequence topological signal enters exclusively through the degree
profiles. 

\subsection{Implementation}
\label{sec:complexity_impl}
Let $N$ be the number of sequences and $k$ the $k$-mer size. The distance matrices require $O(N^2 \cdot J \cdot p^J)$ for $D_p$ and $O(N^2 \cdot 4^k)$ for $D_c$. For each of the $G_p \cdot G_c$ grid points, constructing the VR complex and computing Betti numbers takes $O(N^3)$ in the worst case under the dimension-$2$ simplicial expansion adopted in our implementation, which bounds the simplex count cubically in $N$; without truncation the count grows exponentially in $N$. The total complexity is therefore $O(N^2 (J p^J + 4^k + G_p G_c N))$.

We expand the complex only up to dimension $2$. Across all twelve datasets $\beta_2$ was zero at every grid point, which we attribute to the absence of higher-dimensional voids in sparsely-sampled finite metric spaces; 
the truncation is therefore exact here (dimension-$3$ expansion gives identical Betti numbers; see Section~\ref{sec:results}, Runtime). The bi-filtration grid loop and the $p$-adic distance matrix are parallelised across CPU cores via Joblib. 
\texttt{pVR} is implemented in Python using GUDHI~\cite{maria2014gudhi} (v3.12) for simplicial complexes and persistent homology, scikit-learn (v1.8) for classical classifiers, and XGBoost (v3.2) for gradient-boosted trees. 
All experiments were run on a single workstation with a 12-core AMD Ryzen CPU and 64~GB RAM; no GPU was used.
With the default parameters ($k = 4$, $J = 3$, $G_p = 10$, $G_c = 15$, $N \leq 500$), every dataset completes in under $30$ seconds (Table~\ref{tab:runtime}). 
Code and data for reproducing our experiments are available at: \url{https://github.com/MAHI-Group/pVR}.

\section{Empirical Evaluation}
\label{sec:expts}

\subsection{Setup}
\label{sec:setup}

\smallskip
\noindent\textbf{Datasets.}\enspace We evaluate \texttt{pVR} on twelve genomic classification benchmarks from NCBI GenBank, spanning two scale regimes (Table~\ref{tab:datasets}). The \emph{low-sample regime} comprises six small datasets with $N$ between $28$ and $150$. 
These reflect realistic low-data settings: emerging pathogens, rare species, or newly identified variants for which annotated sequences are scarce. 
The \emph{large-sample regime} comprises six datasets with $N$ between $73$ and $500$, obtained by expanding NCBI search queries on the same organisms, with approximately $75$ to $100$ sequences per class where available. The two regimes use overlapping organisms but distinct sequence sets, enabling a controlled comparison of \texttt{pVR}'s behaviour as sample size grows.

\begin{table}[!htb]
\centering
\caption{Benchmark datasets in two scale regimes. $N$ denotes the number of sequences after filtering, and $C$ denotes the number of classes after merging singletons.}
\label{tab:datasets}
\begin{tabular}{lccl}
\hline
Dataset & $N$ & $C$ & Task \\
\hline
\multicolumn{4}{l}{\textit{Low-sample regime}} \\
Mammalian mito (small) & 30 & 7 & Taxonomic order \\
SARS-CoV-2 (small) & 31 & 5 & Variant lineage \\
HRV (small) & 150 & 3 & Serotype (A/B/C) \\
Influenza HA (small) & 59 & 4 & HA subtype \\
HEV (small) & 29 & 3 & Genotype \\
Ebola (small) & 28 & 5 & Species \\
\hline
\multicolumn{4}{l}{\textit{Large-sample regime}} \\
SARS-CoV-2 (large) & 316 & 4 & Variant lineage \\
Influenza HA (large) & 300 & 4 & HA subtype \\
HRV (large) & 300 & 3 & Serotype \\
HEV (large) & 73 & 3 & Genotype \\
Ebola (large) & 99 & 4 & Species \\
Dengue (large) & 400 & 4 & Serotype \\
\hline
\end{tabular}
\end{table}

All experiments use $k = 4$ and $p = 5$. Classes with fewer than three members are merged into a residual ``Other'' class to ensure that stratified cross-validation is well-defined. Sequences shorter than 100 nucleotides are excluded as outliers. The mammalian dataset spans 11 taxonomic orders, reduced to 7 classes after merging, including Primates, Rodentia, Carnivora, Artiodactyla, Cetacea, and Perissodactyla. The SARS-CoV-2 datasets cover variant lineages (Original, Alpha, Beta, Gamma, Delta, Omicron) with the large-sample variant containing four well-represented lineages after the Alpha query returned no sequences. 
The Ebola dataset spans five species (EBOV, SUDV, BDBV, RESTV, TAFV). The Influenza HA datasets cover four hemagglutinin subtypes (H1N1, H3N2, H5N1, H7N9). The HEV datasets cover four genotypes, with the small dataset containing three after merging. The HRV datasets cover three serotypes (A, B, C). The Dengue dataset, used only at large scale, covers all four DENV serotypes.

\smallskip
\noindent\textbf{Baselines.}\enspace We implement four alignment-free baseline methods on the same data with the same evaluation protocol. FFP-JS uses $k$-mer frequency profiles ($k = 3$) compared via Jensen--Shannon divergence~\cite{sims2009alignment}. 
NVM uses nucleotide-level positional statistics, namely the count, mean position, and
normalised second central moment of each nucleotide~\cite{deng2011novel}; its features
are standardised before computing Euclidean distances. 
Mash uses MinHash sketches with $200$ hash functions and $k = 7$ to estimate Jaccard distance~\cite{ondov2016mash}. The $k$-mer frequency baseline uses $k = 4$ relative-frequency vectors with Euclidean distance after standardisation. Distance-based methods are evaluated using $5$-NN classification, while feature-based methods (the $k$-mer frequency baseline at the feature level, and \texttt{pVR}) are evaluated using XGBoost, RBF-kernel SVM, and $5$-NN. Accuracy is reported as mean $\pm$ standard deviation over $10$ repeats of $5$-fold stratified cross-validation, yielding up to $50$ paired fold accuracies per (dataset, method) cell with seeds $\{42, 43, \ldots, 51\}$. Some small datasets contain a class with fewer than five members, which forces $n_{\mathrm{folds}} < 5$ and reduces the count to $20$--$40$ folds. For headline comparisons, we report a one-sided paired Wilcoxon signed-rank test on the paired fold differences.
We note that fold-level accuracies from repeated cross-validation are not statistically independent, since training sets overlap across folds and seeds. The resulting paired Wilcoxon $p$-values are therefore mildly anti-conservative in the sense of~\cite{Nadeau2003,bouckaert2004evaluating}; we report them as a consistency check on repeated paired differences rather than as formal hypothesis tests, and the gaps we emphasise (Ebola, Influenza HA) here are large enough that this concern does not affect the qualitative conclusions.

\subsection{Results}
\label{sec:results}

We report results in two scale regimes, low-sample and large-sample, followed by a comparison with foundation-model embeddings and supporting analyses.

\smallskip
\noindent\textbf{Low-sample regime.}\enspace
On the six low-sample benchmarks (Table~\ref{tab:main-small}), \texttt{pVR} is strongest
on three (Ebola, Influenza HA, mammalian mitochondrial), is a statistical tie with FFP-JS
on HEV, ties the best baseline on HRV, and trails NVM on SARS-CoV-2.
The largest gain is on Ebola: \texttt{pVR}-SVM reaches $100.0 \pm 0.0\%$ against $78.9 \pm 5.7\%$ for MinHash, a $21.1$-point gap (one-sided paired Wilcoxon, $p < 10^{-4}$, $n = 20$).
Influenza HA is similar ($72.5 \pm 11.0\%$ vs $62.2\%$ for FFP-JS; $p < 10^{-4}$, $n = 50$), as is mammalian mitochondrial ($62.0 \pm 13.7\%$ vs $53.9\%$ for FFP-JS; $p = 0.0002$, $n = 40$). 
The gains track datasets with clearer hierarchical class separation. HEV is a near-tie ($67.0 \pm 8.2\%$ vs $66.0 \pm 6.5\%$ for FFP-JS; $p = 0.27$), and HRV is saturated, with \texttt{pVR}, FFP-JS, and the $k$-mer baseline all at $100\%$. 
The exception is SARS-CoV-2: \texttt{pVR}-XGBoost reaches $42.1 \pm 15.7\%$ against
$47.5 \pm 17.0\%$ for NVM. 
Its variants differ by scattered point mutations rather than hierarchical divergence, which the $p$-adic axis does not capture; raising $k$ from $4$ to $6$ recovers $61.0\%$ (sensitivity analysis below).

\begin{table*}[!htb]
\centering
\caption{Classification accuracy (\%) in the low-sample regime. Each cell is mean $\pm$ std over up to 50 fold accuracies (10 seeds $\times$ 5 folds; reduced when class size forces $n_{\mathrm{folds}} < 5$). Baselines use $5$-NN. \texttt{pVR} reports the best classifier per dataset; the symbol indicates the classifier ($^\dagger$5-NN, $^\ddagger$XGBoost, $^\S$SVM). Best in bold.}
\label{tab:main-small}
\setlength{\tabcolsep}{3pt}
\begin{tabular}{lcccccc}
\hline
Method & Mam. & CoV-2 & HRV & Inf. & HEV & Ebola \\
\hline
FFP-JS       & 53.9 $\pm$ 18.2 & 42.6 $\pm$ 14.7 & 100.0 $\pm$ 0.0 & 62.2 $\pm$ 12.4 & 66.0 $\pm$ 6.5 & 76.8 $\pm$ 6.7 \\
NVM & 46.3 $\pm$ 14.9 & 47.5 $\pm$ 17.0 & 96.7 $\pm$ 2.4 & 61.9 $\pm$ 12.5 & 65.6 $\pm$ 6.7 & 72.5 $\pm$ 2.6 \\
MinHash & 50.4 $\pm$ 14.3 & 25.7 $\pm$ 7.6 & 98.0 $\pm$ 2.2 & 53.9 $\pm$ 13.8 & 57.0 $\pm$ 7.5 & 78.9 $\pm$ 5.7 \\
$k$-mer freq & 49.9 $\pm$ 14.9 & 45.1 $\pm$ 16.8 & 100.0 $\pm$ 0.0 & 60.0 $\pm$ 12.3 & 65.2 $\pm$ 5.6 & 57.9 $\pm$ 2.1 \\
\hline
\texttt{pVR} & \textbf{62.0 $\pm$ 13.7}$^\dagger$
             & 42.1 $\pm$ 15.7$^\ddagger$
             & \textbf{100.0 $\pm$ 0.0}$^\dagger$
             & \textbf{72.5 $\pm$ 11.0}$^\ddagger$
             & \textbf{67.0 $\pm$ 8.2}$^\S$
             & \textbf{100.0 $\pm$ 0.0}$^\S$ \\
\hline
\end{tabular}
\end{table*}

\smallskip
\noindent\textbf{Large-sample regime.}\enspace
With more data, the benchmarks saturate (Table~\ref{tab:main-large}): most methods land between $96$ and $100\%$, and Dengue, Ebola, and HRV exceed $99\%$ across every baseline, so simple compositional features already suffice. \texttt{pVR} stays in this band, with \texttt{pVR}-SVM at $99.2$--$100\%$ on five datasets and $98.5 \pm 3.1\%$ on HEV-large; no gap to the best baseline is significant ($p > 0.16$). The lack of an edge here reflects task saturation, not an uninformative bi-filtration: its hierarchical structure may still help on downstream tasks such as phylogenetic reconstruction or recombination analysis (Section~\ref{sec:discussion}).

\begin{table*}[!htb]
\centering
\caption{Classification accuracy (\%) in the large-sample regime. Each cell is mean $\pm$ std over up to $50$ fold accuracies (10 seeds $\times$ 5 folds). Best in bold.}
\label{tab:main-large}
\setlength{\tabcolsep}{3pt}
\begin{tabular}{lcccccc}
\hline
Method & Dengue & Ebola & HEV & HRV & Inf. & CoV-2 \\
\hline
FFP-JS       & \textbf{100.0 $\pm$ 0.0} & \textbf{100.0 $\pm$ 0.0} & 97.3 $\pm$ 1.5 & \textbf{99.4 $\pm$ 1.4} & 99.0 $\pm$ 1.2 & 99.2 $\pm$ 1.0 \\
NVM & 99.7 $\pm$ 0.5 & 100.0 $\pm$ 0.0 & 95.6 $\pm$ 2.0 & 96.3 $\pm$ 2.3 & 99.0 $\pm$ 1.2 & 96.1 $\pm$ 2.6 \\
MinHash & 100.0 $\pm$ 0.0 & 100.0 $\pm$ 0.0 & 98.6 $\pm$ 1.4 & 94.9 $\pm$ 2.7 & 99.7 $\pm$ 0.7 & 70.9 $\pm$ 3.7 \\
$k$-mer freq & \textbf{100.0 $\pm$ 0.0} & \textbf{100.0 $\pm$ 0.0} & \textbf{99.9 $\pm$ 0.6} & \textbf{99.4 $\pm$ 1.4} & 98.8 $\pm$ 1.3 & 99.4 $\pm$ 0.9 \\
\hline
\texttt{pVR} & 99.4 $\pm$ 0.8$^\ddagger$
             & \textbf{100.0 $\pm$ 0.0}$^\dagger$
             & 98.5 $\pm$ 3.1$^\S$
             & 99.2 $\pm$ 1.1$^\S$
             & 99.5 $\pm$ 0.9$^\S$
             & \textbf{99.6 $\pm$ 0.8}$^\S$ \\
\hline
\end{tabular}
\end{table*}

\smallskip
\noindent\textbf{Comparison with NT embeddings.}\enspace
We also compared \texttt{pVR} against zero-shot embeddings from Nucleotide Transformer~v2 (NT~v2, 500M parameters, multi-species)~\cite{dalla2025nucleotide} on Ebola, mammalian mitochondrial, and Influenza HA. We mean-pool the final-layer hidden states (chunking sequences past the 2048-token limit into non-overlapping 12~kb windows and averaging) and use them as frozen features under the same cross-validation protocol. \texttt{pVR} outperforms NT v2 on all three (Table~\ref{tab:nt_compare}), by $11.4$ points on
Ebola, $7.1$ on mammalian, and $6.7$ on Influenza~HA, and its cosine-UMAP projection
separates subtypes more cleanly than NT v2 (Figure~\ref{fig:nt_umap}). The comparison is
deliberately narrow, using frozen embeddings without fine-tuning and at most $60$
labelled sequences per task. 
In that regime \texttt{pVR}'s handcrafted hierarchical and compositional structure is a stronger inductive bias~\cite{mitchell1997machine,baxter2000model} than generic pretraining; a comparison against fine-tuned foundation models at larger $N$ is left to future work.

\begin{table}[!htb]
\centering
\caption{Comparison of \texttt{pVR} with Nucleotide Transformer~v2 (500M, multi-species) zero-shot embeddings on three low-sample benchmarks. NT embeddings are mean-pooled hidden states from the final layer; long sequences are chunked and averaged. Both methods use repeated stratified CV (up to $50$ folds) and the best-performing classifier per row.}
\label{tab:nt_compare}
\setlength{\tabcolsep}{4pt}
\begin{tabular}{lcc}
\hline
Dataset & \texttt{pVR} (best) & NT v2 (best probe) \\
\hline
Ebola & 100.0 $\pm$ 0.0$^\S$ & 88.6 $\pm$ 9.9$^\S$ \\
Mammalian & 62.0 $\pm$ 13.7$^\dagger$ & 54.9 $\pm$ 15.2$^\S$ \\
Influenza HA & 72.5 $\pm$ 11.0$^\ddagger$ & 65.8 $\pm$ 8.8$^\S$ \\
\hline
\end{tabular}
\end{table}

\begin{figure}[!htb]
\centering
\includegraphics[width=\columnwidth]{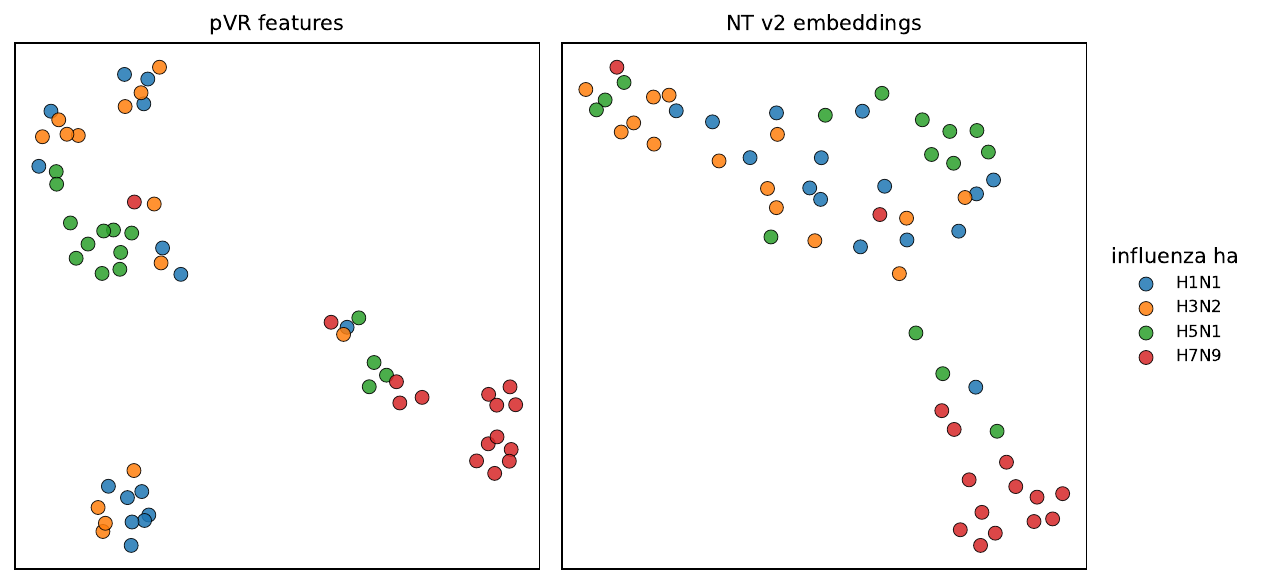}
\caption{Cosine-UMAP projections of \texttt{pVR} features (left) and
Nucleotide Transformer v2 zero-shot embeddings (right) on
Influenza HA-small ($N=59$, four subtypes). \texttt{pVR} produces
visibly subtype-separated clusters; NT~v2 embeddings show only weak
separation. Equivalent figures for the other two benchmarks are
included within our code repository.}
\label{fig:nt_umap}
\end{figure}

\smallskip
\noindent\textbf{Ablation.}\enspace
We classify with one feature group at a time, using XGBoost throughout (Tables~\ref{tab:ablation-small} and~\ref{tab:ablation-large}). The relative contribution of each group varies sharply across datasets. Combining axes helps most on Ebola, where the full representation reaches $91.4 \pm 8.0\%$ versus $90.7\%$ for the best single group ($p$-adic histograms), rising to $100\%$ under SVM (Table~\ref{tab:main-small}). On SARS-CoV-2 the combined representation ($42.1 \pm 15.7\%$) edges past every component (best single: $k$-mer frequencies, $40.7\%$) but stays below NVM. HRV is saturated across all ablations, and on Influenza HA the combination matches its best component. HEV and mammalian instead expose a classifier effect: under XGBoost the combination underperforms its best single group ($55.6\%$ vs $61.1\%$ on HEV; $47.7\%$ on mammalian), yet the same features reach $67.0\%$ and $62.0\%$ under SVM and $5$-NN respectively, so the loss reflects XGBoost's sensitivity to noisy features rather than a weakness of the representation.

At large $N$ most groups already exceed $95\%$, leaving little to combine; Dengue, Ebola, HRV, and Influenza are near-saturated on both hierarchical and compositional features alone. The informative exceptions are HEV-large, which seems to repeat the XGBoost pattern ($p$-adic histograms $93.5\%$, combined $92.4\%$, against $98.5\%$ under SVM/$5$-NN), and SARS-CoV-2-large, where the $p$-adic VR degree profiles alone reach only $80.0\%$ but lift the combination to $98.0\%$, complementing the $k$-mer frequencies that carry most of the signal on this dataset.

\begin{table*}[!htb]
\centering
\caption{Low-sample ablation study (XGBoost mean $\pm$ std accuracy \%, repeated CV with $10$ seeds $\times$ $5$ folds, matching the main-table protocol). Each row reports a classifier trained on a single feature group; \texttt{pVR} combined uses all groups. Comp.\ VR denotes the compositional ($L_1$) VR degree profiles.}
\label{tab:ablation-small}
\setlength{\tabcolsep}{3pt}
\begin{tabular}{lcccccc}
\hline
Features & Mam. & CoV-2 & HRV & Inf. & HEV & Ebola \\
\hline
$p$-adic VR    & 49.6 $\pm$ 15.4 & 39.5 $\pm$ 16.5 & 97.4 $\pm$ 2.5 & 64.1 $\pm$ 14.1 & 59.3 $\pm$ 11.4 & 85.4 $\pm$ 8.9 \\
Comp.\ VR      & \textbf{56.5 $\pm$ 17.2} & 26.4 $\pm$ 7.5 & 96.9 $\pm$ 3.1 & 67.4 $\pm$ 12.7 & \textbf{61.1 $\pm$ 11.0} & 86.1 $\pm$ 6.2 \\
Bi-filt.\ topo & 49.0 $\pm$ 16.1 & 39.5 $\pm$ 16.5 & 98.5 $\pm$ 2.2 & 64.4 $\pm$ 13.4 & 60.6 $\pm$ 10.4 & 85.0 $\pm$ 7.8 \\
$p$-adic hist  & 54.6 $\pm$ 15.4 & 37.5 $\pm$ 15.3 & 98.7 $\pm$ 2.2 & \textbf{73.0 $\pm$ 12.3} & 54.9 $\pm$ 10.5 & 90.7 $\pm$ 9.3 \\
$k$-mer freq   & 43.3 $\pm$ 14.2 & 40.7 $\pm$ 15.6 & 97.9 $\pm$ 2.4 & 72.7 $\pm$ 10.8 & 53.7 $\pm$ 15.5 & 87.5 $\pm$ 8.7 \\
\hline
\texttt{pVR} combined & 47.7 $\pm$ 13.5 & \textbf{42.1 $\pm$ 15.7} & \textbf{98.3 $\pm$ 2.5} & 72.5 $\pm$ 11.0 & 55.6 $\pm$ 15.7 & \textbf{91.4 $\pm$ 8.0} \\
\hline
\end{tabular}
\end{table*}

\begin{table*}[!htb]
\centering
\caption{Large-sample ablation study (XGBoost mean $\pm$ std accuracy \%, repeated CV with $10$ seeds $\times$ $5$ folds). Feature conventions are identical to those in Table~\ref{tab:ablation-small}.}
\label{tab:ablation-large}
\setlength{\tabcolsep}{3pt}
\begin{tabular}{lcccccc}
\hline
Features & Dengue & Ebola & HEV & HRV & Inf. & CoV-2 \\
\hline
$p$-adic VR    & 98.1 $\pm$ 1.5 & 98.6 $\pm$ 2.3 & 88.8 $\pm$ 4.6 & 95.0 $\pm$ 2.4 & 96.1 $\pm$ 2.4 & 80.0 $\pm$ 3.4 \\
Comp.\ VR      & 98.5 $\pm$ 1.0 & 98.3 $\pm$ 3.0 & 90.3 $\pm$ 5.7 & 95.3 $\pm$ 2.3 & 96.9 $\pm$ 2.3 & 93.6 $\pm$ 2.9 \\
Bi-filt.\ topo & 99.4 $\pm$ 0.8 & 99.4 $\pm$ 1.9 & 89.4 $\pm$ 5.2 & 95.5 $\pm$ 2.4 & 98.5 $\pm$ 1.4 & 93.2 $\pm$ 2.9 \\
$p$-adic hist  & \textbf{99.6 $\pm$ 0.8} & \textbf{100.0 $\pm$ 0.0} & \textbf{93.5 $\pm$ 5.3} & \textbf{98.8 $\pm$ 1.3} & 98.4 $\pm$ 1.6 & \textbf{98.1 $\pm$ 1.6} \\
$k$-mer freq   & \textbf{99.4 $\pm$ 0.9} & 97.6 $\pm$ 3.1 & 91.0 $\pm$ 5.5 & \textbf{99.0 $\pm$ 1.3} & \textbf{98.5 $\pm$ 1.6} & \textbf{98.3 $\pm$ 1.7} \\
\hline
\texttt{pVR} combined & \textbf{99.4 $\pm$ 0.8} & 99.6 $\pm$ 1.7 & 92.4 $\pm$ 6.3 & 98.5 $\pm$ 1.5 & \textbf{98.6 $\pm$ 1.6} & 98.0 $\pm$ 1.9 \\
\hline
\end{tabular}
\end{table*}

\smallskip
\noindent\textbf{Feature importance.}\enspace
Table~\ref{tab:importance} reports normalised XGBoost gain per feature group across the
three per-sequence groups. The global Betti numbers $\beta_0, \beta_1$ take a single
value per dataset at each grid point, so they are constant across sequences and cannot
serve as per-sequence features by construction; the topological signal enters instead
through the degree profiles, which summarise each sequence's connectivity in the
bi-filtered complex.
Degree profiles lead where classes are well separated (HRV-small, $54.7\%$; Ebola-large, $28.7\%$); $p$-adic histograms contribute steadily across both regimes; and $k$-mer frequencies dominate when class structure is compositional, most starkly on SARS-CoV-2-large at $95.7\%$. 
The three groups are thus complementary, with their relative weight set by the dataset's
evolutionary structure, so no single group dominates across all datasets.

\begin{table}[!htb]
\centering
\caption{Feature group importance (\%, XGBoost gain). Deg.\ denotes degree profiles, Hist denotes $p$-adic histograms, and Freq denotes $k$-mer frequencies. The global Betti numbers $\beta_0, \beta_1$ are constant across sequences within a dataset; they are omitted.}
\label{tab:importance}
\begin{tabular}{lccc}
\hline
Dataset & Deg. & Hist & Freq \\
\hline
\multicolumn{4}{l}{\textit{Low-sample}} \\
Ebola & 10.5 & 45.5 & 44.0 \\
HEV & 13.1 & 26.5 & 60.4 \\
HRV & 54.7 & 27.6 & 17.7 \\
Influenza & 1.6 & 35.4 & 63.0 \\
Mammalian & 2.4 & 44.7 & 52.9 \\
SARS-CoV-2 & 3.4 & 30.8 & 65.8 \\
\hline
\multicolumn{4}{l}{\textit{Large-sample}} \\
Dengue & 23.6 & 49.8 & 26.6 \\
Ebola & 28.7 & 26.4 & 45.0 \\
HEV & 10.8 & 43.7 & 45.5 \\
HRV & 10.5 & 81.9 & 7.5 \\
Influenza & 15.4 & 50.4 & 34.2 \\
SARS-CoV-2 & 0.4 & 3.9 & 95.7 \\
\hline
\end{tabular}
\end{table}

\smallskip
\noindent\textbf{Runtime.}\enspace
Every dataset runs in under $30$ seconds on a $12$-core workstation (Table~\ref{tab:runtime}). The costliest is Dengue-large ($N=400$): about $3$ seconds for the distance matrices and $25$ for the bi-filtration grid loop over $G_p \times G_c = 150$ threshold pairs. Restricting simplicial expansion to dimension $2$, exact here since $\beta_2$ was zero on every dataset, holds peak memory below $5$~GB; dimension-$3$ expansion pushed it past $60$~GB on Dengue-large for identical features. Cost is set less by $N$ than by the density of the complexes across the grid, and only a small fraction of grid points hit the expensive intermediate-density regime.

\begin{table}[!htb]
\centering
\caption{Runtime in seconds. Column ``pVR dist.'' is the time to compute $D_p$ and $D_c$. Column ``pVR feat.'' is the time to compute the bi-filtration grid and extract features. Column ``Baselines'' is the total time to compute all four baseline distance matrices.}
\label{tab:runtime}
\begin{tabular}{lrrrr}
\hline
Dataset & $N$ & pVR dist. & pVR feat. & Baselines \\
\hline
\multicolumn{5}{l}{\textit{Low-sample}} \\
Mammalian & 30 & 0.3 & 0.3 & 0.8 \\
SARS-CoV-2 & 31 & 0.5 & 0.5 & 1.0 \\
HRV & 150 & 0.8 & 2.0 & 2.3 \\
Influenza & 59 & 0.1 & 0.3 & 0.2 \\
HEV & 29 & 0.1 & 0.2 & 0.4 \\
Ebola & 28 & 0.3 & 0.3 & 0.8 \\
\hline
\multicolumn{5}{l}{\textit{Large-sample}} \\
Dengue & 400 & 3.1 & 25.3 & 9.1 \\
Ebola & 99 & 1.2 & 1.2 & 3.2 \\
HEV & 73 & 0.3 & 0.4 & 1.2 \\
HRV & 300 & 1.6 & 12.7 & 5.1 \\
Influenza & 300 & 0.7 & 8.2 & 1.9 \\
SARS-CoV-2 & 316 & 6.0 & 21.4 & 13.0 \\
\hline
\end{tabular}
\end{table}

\smallskip
\noindent\textbf{Sensitivity to hyperparameters.}\enspace
Table~\ref{tab:sens} reports XGBoost accuracy across $k \in \{3,4,5,6\}$; by Proposition~\ref{prop:prime_invariance} the prime is irrelevant ($p \in \{5,7,11,13\}$ give identical features), so only $k$ is varied. At low $N$, most datasets peak at $k=4$ or $5$, with one informative exception: SARS-CoV-2-small jumps from $47.6\%$ at $k=4$ to $61.0\%$ at $k=6$, consistent with longer $k$-mers capturing its variant-defining point mutations. The smallest datasets move the other way; mammalian and Influenza HA degrade at $k=6$, where the $4^6 = 4096$-dimensional compositional space is hard to estimate from a few dozen sequences. We fix $k=4$ for all main experiments rather than tuning per dataset, which would overfit and break cross-method comparability. 
At large $N$, accuracy is flat across $k$ to within a few points, so the choice of $k$-mer length barely matters.

\begin{table*}[!htb]
\centering
\caption{Sensitivity to $k$-mer size: XGBoost accuracy (\%, mean $\pm$ std over a single 5-fold split). The prime $p$ is invariant by Proposition~\ref{prop:prime_invariance}.}
\label{tab:sens}
\setlength{\tabcolsep}{3pt}
\begin{tabular}{lcccccc}
\hline
$k$ & \multicolumn{6}{c}{Low-sample datasets} \\
\cline{2-7}
 & Mam. & CoV-2 & HRV & Inf. & HEV & Ebola \\
\hline
3 & 53.6 $\pm$ 18.0 & 45.7 $\pm$ 20.2 & 98.7 $\pm$ 1.6 & 69.7 $\pm$ 15.2 & 51.9 $\pm$ 5.2 & 89.3 $\pm$ 10.7 \\
4 & 50.4 $\pm$ 14.9 & 47.6 $\pm$ 17.6 & 99.3 $\pm$ 1.3 & 73.0 $\pm$ 7.8 & 62.1 $\pm$ 2.1 & 92.9 $\pm$ 7.1 \\
5 & 56.2 $\pm$ 11.9 & 44.8 $\pm$ 17.8 & 100.0 $\pm$ 0.0 & 72.9 $\pm$ 11.0 & 58.3 $\pm$ 8.3 & 89.3 $\pm$ 10.7 \\
6 & 39.7 $\pm$ 23.0 & 61.0 $\pm$ 13.9 & 99.3 $\pm$ 1.3 & 64.5 $\pm$ 9.3 & 58.6 $\pm$ 1.4 & 92.9 $\pm$ 7.1 \\
\hline
$k$ & \multicolumn{6}{c}{Large-sample datasets} \\
\cline{2-7}
 & Dengue & Ebola & HEV & HRV & Inf. & CoV-2 \\
\hline
3 & 99.5 $\pm$ 0.6 & 99.0 $\pm$ 1.7 & 83.3 $\pm$ 8.3 & 99.0 $\pm$ 1.3 & 98.7 $\pm$ 1.2 & 98.7 $\pm$ 0.6 \\
4 & 99.0 $\pm$ 0.9 & 99.0 $\pm$ 1.7 & 84.7 $\pm$ 9.7 & 98.7 $\pm$ 1.2 & 99.0 $\pm$ 0.8 & 97.8 $\pm$ 2.4 \\
5 & 98.5 $\pm$ 0.9 & 100.0 $\pm$ 0.0 & 84.7 $\pm$ 9.7 & 98.7 $\pm$ 1.9 & 97.7 $\pm$ 2.5 & 97.8 $\pm$ 1.9 \\
6 & 99.0 $\pm$ 0.9 & 99.0 $\pm$ 1.7 & 80.6 $\pm$ 5.6 & 98.7 $\pm$ 1.2 & 97.7 $\pm$ 2.5 & 98.7 $\pm$ 1.2 \\
\hline
\end{tabular}
\end{table*}

\smallskip
\noindent\textbf{Topological visualisation.}\enspace
Figure~\ref{fig:topology_both} shows the $p$-adic and compositional $L_1$ distance matrices with the $\beta_0, \beta_1$ heatmaps for two datasets. The $p$-adic matrix has sharp block structure aligned with class boundaries (taxonomic orders for mammalian, serotypes for HRV), while the compositional matrix is smoother and more graded. The $\beta_1$ heatmaps make Proposition~\ref{prop:nontrivial} concrete on real data: nontrivial $1$-cycles appear in a narrow band at low-to-intermediate $p$-adic and moderate compositional threshold. At very low $p$-adic threshold the complex is restricted to within-clade pairs, and the compositional axis decides which connect; at high $p$-adic threshold it collapses to a full simplex per component and $\beta_1 \to 0$, as Theorem~\ref{thm:trivial} predicts. The band's shape varies with the data, wider and diffuse for mammalian (seven orders), narrower for HRV (three well-separated serotypes), sparse with low peaks for SARS-CoV-2-small (near-uniform $k$-mer composition across variants), and peaked at intermediate thresholds on small Ebola (cross-species relationships among \emph{Ebolavirus}).
We provide the heatmaps for SARS-CoV-2 and Ebola in the code repository.

\begin{figure*}[!htb]
\centering
\begin{tabular}{cc}
\includegraphics[width=0.48\textwidth]{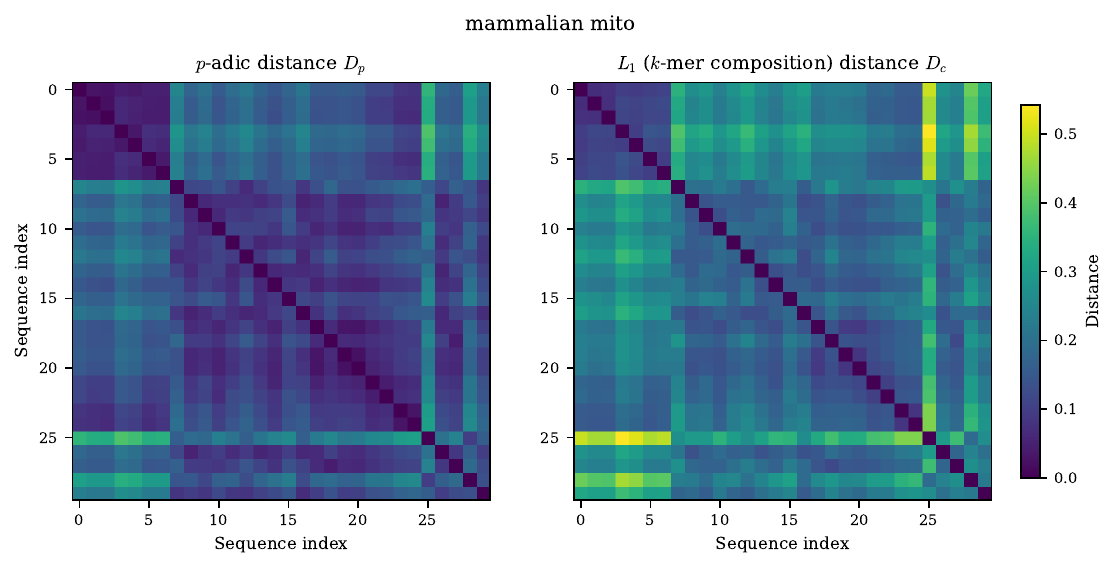} &
\includegraphics[width=0.48\textwidth]{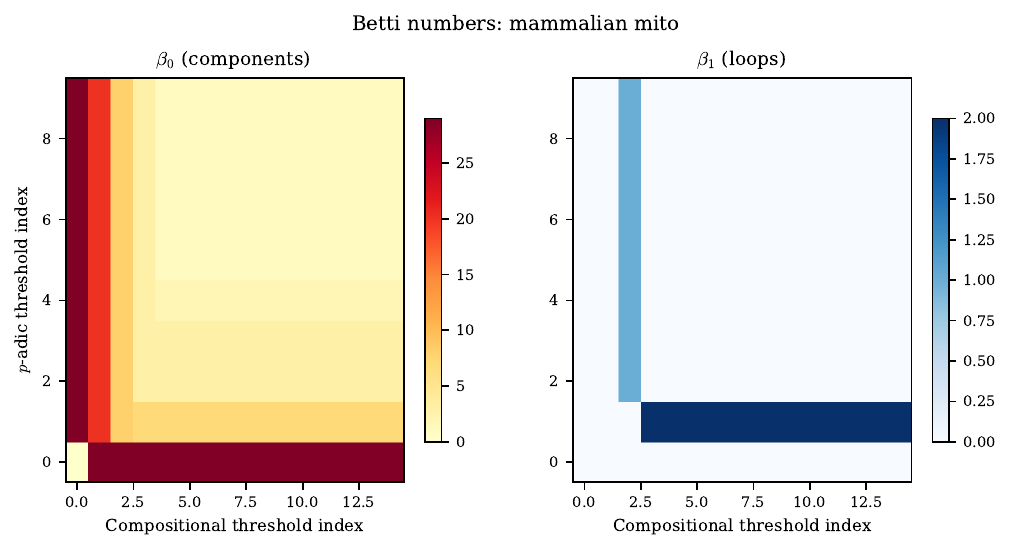} \\
\multicolumn{2}{c}{\small (a) Mammalian mitochondrial ($N=30$, 7 taxonomic orders)} \\[6pt]
\includegraphics[width=0.48\textwidth]{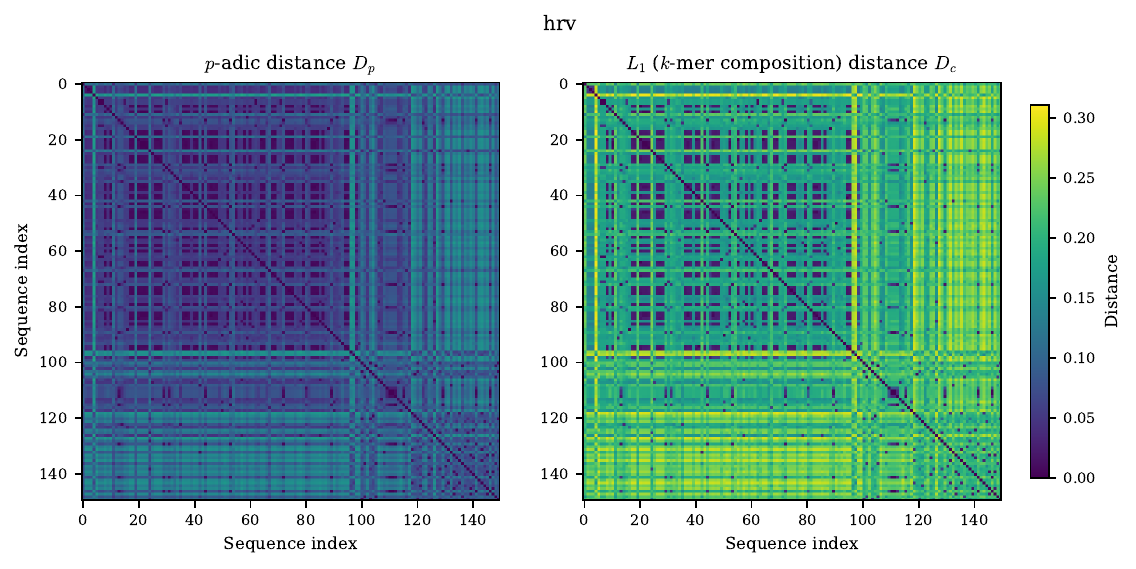} &
\includegraphics[width=0.48\textwidth]{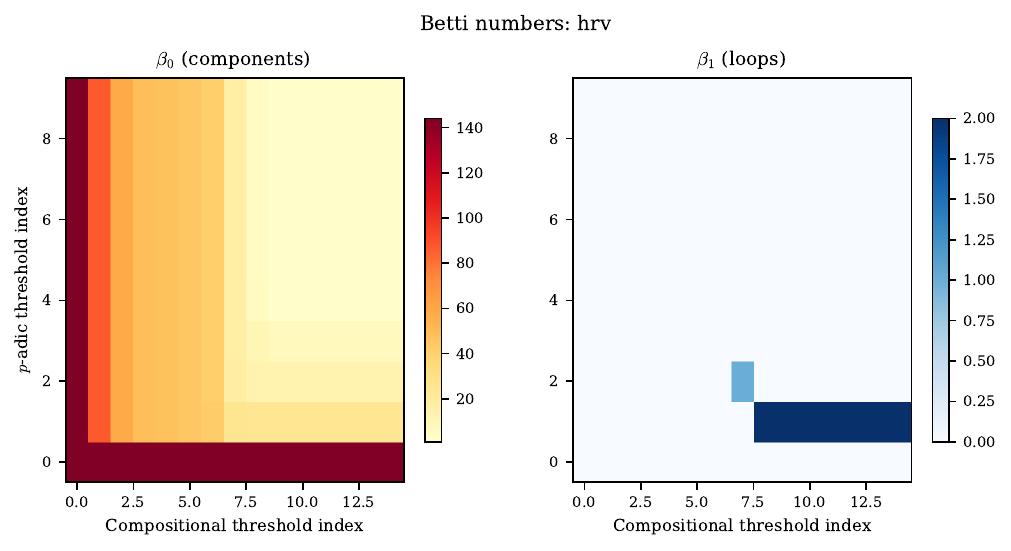} \\
\multicolumn{2}{c}{\small (b) HRV ($N=150$, three serotypes)} \\
\end{tabular}
\caption{Distance matrices (left) and Betti heatmaps (right) for two
low-sample datasets. \textbf{Distance matrices:} the $p$-adic distance
$D_p$ exhibits sharp block structure aligned with class boundaries while
the compositional $L_1$ distance $D_c$ varies more smoothly; both panels
in each row share a common colour scale. \textbf{Betti heatmaps:}
$\beta_0$ decays as both thresholds grow, while $\beta_1$ becomes
nontrivial in a narrow band at moderate compositional threshold and low
to intermediate $p$-adic threshold, empirically realising the
bi-filtration cycle predicted by Proposition~\ref{prop:nontrivial}.
$\beta_2$ vanished at every grid point and is therefore omitted.}
\label{fig:topology_both}
\end{figure*}

\section{Discussion}
\label{sec:discussion}

\smallskip
\noindent\textbf{When pVR helps.}\enspace
In the low-sample regime, three of six datasets improve significantly over baselines
(paired Wilcoxon $p < 0.05$), by $8.1$ to $21.1$ points, with a fourth (HEV) a
non-significant gain. 
The bi-filtration acts as an inductive prior~\cite{mitchell1997machine,kontolati2025biology}, constraining the hypothesis space with structural information that matters most when labels are scarce. 
This places \texttt{pVR} alongside efforts to build structured prior knowledge into
learning systems, whether supplied as external domain knowledge~\cite{dash2022review} or
mined from the data~\cite{dash2026birdnet}.
At large $N$ the alignment-free methods all converge to near-perfect accuracy, and pVR neither beats nor trails $k$-mer frequencies.
This saturation reflects the easiness of the task rather than a limit of the method, and
the stability and prime-invariance guarantees hold across both regimes.
In short, \texttt{pVR} helps when data is scarce and divergence is hierarchical, the
setting of emerging pathogens, rare species, and newly identified variants.
At large $N$ the task itself saturates: three datasets reach $99$--$100\%$ for every method and the rest exceed $95\%$, 
the one exception being MinHash on SARS-CoV-2-large ($70.9\%$), a known weakness of
sketch methods on near-identical genomes. 
The hierarchical signal from bi-filtration is more likely to help on harder tasks that do not saturate, such as phylogenetic reconstruction~\cite{semple2003phylogenetics}.

\smallskip
\noindent\textbf{When it does not.}\enspace
\texttt{pVR} underperforms on SARS-CoV-2-small ($42.1 \pm 15.7\%$ against $47.5 \pm 17.0\%$
for NVM), and the ablation is consistent with this: the $p$-adic VR and histogram
components score $39.5\%$ and $37.5\%$, both under the $k$-mer baseline. SARS-CoV-2
lineages descend from one ancestor by scattered point mutations rather than deep clade
divergence, so the $p$-adic prefix ordering has little to exploit.
Raising $k$ from $4$ to $6$ lifts accuracy from $47.6\%$ to $61.0\%$ (single-seed CV; sensitivity analysis above), as longer prefixes catch mutations that shorter ones miss; we keep $k = 4$ everywhere to avoid per-dataset tuning, which would overfit. Lineages dominated by point substitution call for longer $k$ or a different encoding. HEV-small is the boundary case, a statistical tie with FFP-JS ($67.0$ vs $66.0$, $p = 0.27$): HEV genotypes recombine~\cite{smith2014consensus}, breaking the hierarchical assumption, consistent with earlier TDA findings on recombination~\cite{camara2017topological}.

\smallskip
\noindent\textbf{Practical considerations.}\enspace
RBF-SVM is the most reliable \texttt{pVR} classifier, never ranking worst and ranking
best on the largest number of datasets, including the small ones where the gains are
largest. This may be due to its smooth decision boundary, which can tolerate noisy
features.
On the other hand, XGBoost~\cite{chen2016xgboost} is preferable when the class margin sits in a few features, as on SARS-CoV-2-large, where $k$-mer frequencies carry $95.7\%$ of the XGBoost gain. 
Similarly, $5$-NN suits datasets where degree profiles align with class structure, as on mammalian, where neighbourhood structure tracks taxonomy. 
We recommend SVM as the default among the classifiers tested here. 
A natural next step is to feed the bi-filtration features into a deep neural network, which could model richer, nonlinear interactions among the degree profiles, $p$-adic histograms, and $k$-mer frequencies in its hidden layers than a kernel method allows. The constraint is data: with $N$ in the tens to low hundreds, such a model would need strong regularisation or pretraining to avoid overfitting, so the gain is likeliest in the large-sample regime or on the downstream tasks discussed above.

Feature importance follows the same logic (Table~\ref{tab:importance}): degree profiles dominate when classes are well separated (HRV-small $54.7\%$, Ebola-large $28.7\%$), $p$-adic histograms when hierarchy aligns with classes, and $k$-mer frequencies when the structure is compositional. 
The Betti numbers are dataset-level descriptors rather than per-sequence features, which
points to an extension we leave open: per-sequence topological features, such as local
persistent homology around each vertex or on vertex-removed subcomplexes.

\section{Concluding Remarks}
\label{sec:conclusion}

\texttt{pVR} bridges $p$-adic number theory and topological data
analysis for alignment-free genomic classification. The construction
rests on four mathematical observations: ultrametric VR complexes have
trivial higher homology, the bi-filtration recovers nontrivial topology
from the interaction of two metrics, the resulting persistence module
is stable under metric perturbations, and the choice of prime is
immaterial provided $p \geq |\Sigma| + 1$. 
Empirically, \texttt{pVR} outperforms four alignment-free baselines on three of six
low-sample benchmarks (significance on Ebola, Influenza HA, and the mammalian dataset;
gap up to $21.1$ percentage points) and remains competitive when all methods saturate at
larger sample sizes.

We note several limitations of our present work. 
SARS-CoV-2 variants derive
from a single ancestral genome by scattered point mutations rather
than deep clade divergence; 
on this benchmark \texttt{pVR} underperforms compositional baselines by $5.4$ points. 
The HEV genotypes, which are
known to recombine~\cite{smith2014consensus}, sit at the boundary, with
\texttt{pVR} matching FFP-JS within fold-level noise. The $p$-adic axis
contributes when divergence is hierarchical and contributes little when
it is not.

The $O(N^3)$ cost of the bi-filtration grid loop becomes the bottleneck beyond $N = 500$ sequences, 
and witness-complex approximations on landmark sequences~\cite{de2004topological} are the standard route to larger $N$. 
The present uniform $G_p \times G_c$ grid could be replaced by an
adaptive grid concentrated where the topology changes. 
The nucleotide-to-digit assignment used in our work 
($A \mapsto 1, C \mapsto 2, G \mapsto 3, T \mapsto 4$) 
is arbitrary up to permutation. 
A biochemically
informed assignment that groups purines and pyrimidines may yield
sharper hierarchical structure than the lexicographic mapping. 
The hand-designed degree profiles could give way to learnable vectorisations of the bi-persistence module, such as persistence images~\cite{adams2017persistence} or neural persistence layers~\cite{hofer2019learning}, feeding into a deep classifier in place of the kernel methods used here. 
Whether the bi-filtration is more
informative for phylogenetic reconstruction than for classification
remains an open empirical question.

\section*{Acknowledgement}
We used Anthropic's Claude (Opus 4.x) to help draft portions of the Related Work, debug
the implementation, interpret results, and condense the Results and Discussion. All
proofs, code, numerical results, and claims were verified by the authors, who take full
responsibility for the content.


\bibliographystyle{IEEEtran}
\bibliography{refs}

\end{document}